\documentclass{article}
\usepackage{kr}
\usepackage{times}
\usepackage{boondox-cal}
\usepackage[bb=boondox]{mathalfa}
\usepackage{xfrac}
\usepackage{stmaryrd}
\usepackage{iftex}

\usepackage{bm} %bold math symbols \bm{}
\usepackage{soul}
\usepackage{url}
\usepackage[shortlabels]{enumitem}
\usepackage[hidelinks]{hyperref}
\usepackage[T1]{fontenc}
\usepackage[utf8]{inputenc}
\usepackage[small]{caption}
\usepackage{graphicx}
\usepackage{mathtools}
\usepackage{amsmath}
\usepackage{amssymb}
\usepackage{amsthm}
\usepackage{booktabs}
\usepackage{algorithm}
\usepackage{algorithmic}
\usepackage{amsfonts}
\usepackage{xcolor}
\usepackage{multicol}
\usepackage{multirow}
\usepackage{verbatim}
\usepackage{thm-restate}
\usepackage[normalem]{ulem}
\usepackage[all]{xy}
\newcommand{\onehalf}{\overline{\tfrac{1}{2}}}
\newcommand{\one}{\overline{1}}
\newcommand{\zero}{\overline{0}}

\newcommand{\Var}{\mathtt{Var}}

\newcommand{\Prob}{\mathsf{Pr}}

\newcommand{\eprplus}{\mathsf{EPR}^+}
\newcommand{\emrplus}{\mathsf{EMR}^+}

\newcommand{\umodrule}{\mathsf{UMR}}
\newcommand{\uprobrule}{\mathsf{UPR}}
\newcommand{\uml}{\mathsf{UML}}
\newcommand{\upl}{\mathsf{UPL}}

\newcommand{\Luk}{{\normalfont{\textbf{{\L}}}}}
\newcommand{\LukProd}{\Luk\mathbf{\Pi}}
\newcommand{\KL}{{\normalfont\textbf{K{\L}}}}

\newcommand{\KLukProdProb}{\KL\bm{\Pi}\mathbf{\tfrac{1}{2}}^\mathsf{Pr}}

\newcommand{\KLukProd}{\KL\bm{\Pi}\onehalf}

\newcommand{\fb}{\mathsf{fb}}
\newcommand{\sfive}{\mathbf{S5}}

\newcommand{\prea}{\mathsf{EA^{Pr}}}
\newcommand{\preafb}{\mathsf{EA^{Pr}_{fb}}}
\newcommand{\EALukProd}{\mathsf{EA}^{\Luk\mathbf{\Pi}}}
\newcommand{\EALukProdfb}{\mathsf{EA}^{\Luk\mathbf{\Pi}}_{\mathsf{fb}}}
\newcommand{\EA}{\mathsf{EA}}

\newcommand{\Imbb}{\mathbb{I}}
\newcommand{\Nmbb}{\mathbb{N}}

\newcommand{\Smbb}{\mathbb{S}}

\newcommand{\Cmsf}{\mathsf{C}}

\newcommand{\Emsf}{\mathsf{E}}

\newcommand{\Lmsf}{\mathsf{L}}

\newcommand{\Rmsf}{\mathsf{R}}
\newcommand{\Smsf}{\mathsf{S}}

\newcommand{\Wmsf}{\mathsf{W}}

\newcommand{\Bmc}{\mathcal{B}}

\newcommand{\Fmc}{\mathcal{F}}

\newcommand{\Imc}{\mathcal{I}}

\newcommand{\Kmc}{\mathcal{K}}

\newcommand{\Mmc}{{\mathcal{M}}}
\newcommand{\Nmc}{{\mathcal{N}}}

\newcommand{\Smc}{{\mathcal{S}}}
\newcommand{\Tmc}{{\mathcal{T}}}

\newcommand{\lmc}{\mathcal{l}}

\newcommand{\Efrak}{\mathfrak{E}}
\newcommand{\Ffrak}{\mathfrak{F}}

\newcommand{\Mfrak}{{\mathfrak{M}}}

\newcommand{\Xfrak}{{\mathfrak{X}}}

\newcommand{\Fmbf}{\mathbf{F}}
\newcommand{\Kmbf}{\mathbf{K}}
\newcommand{\Lmbf}{\mathbf{L}}

\newcommand{\pspace}{\mathsf{PSPACE}}

\newcommand{\TEALukProd}{\Tmc\!(\EALukProd_\fb)}

\newcommand{\TKLukProd}{\Tmc\!(\KLukProd_\fb)}

\newcommand{\atom}{\mathsf{At}}

\newcommand{\real}{\mathsf{rl}}

\newcommand{\ClusterSet}{\mathfrak{Cl}}
\newcommand{\cluster}{\mathfrak{cl}}
\newcommand{\ActionSet}{\mathfrak{Ac}}

\newcommand\myoverset[2]{\overset{\scriptstyle #1\mathstrut}{\scriptstyle #2\mathstrut}}
\newcommand{\subformulas}{\mathsf{SF}}

\newcommand{\Ac}{\mathtt{Act}} % the set of basic actions
\newcommand{\Ag}{\mathtt{Agt}} % the set of agents
\newcommand{\Lev}{\mathcal{L}_{\mathbf{evt}}} % event language; multi-agent S5 + actions (K)
\newcommand{\Lpr}{\mathcal{L}^{\Prob}_\mathsf{EA}} % (full) probabilistic language
\newcommand{\Lmod}{\mathcal{L}_\mathsf{EA}} % the language of (S5+K)LukProd
\newcommand{\sneg}{\mathord{\sim}} % Boolean negation in event formulas

\newcommand{\knows}{\mathsf{K}} % a proposal for the knowledge modality
\newcommand{\ActBox}[1]{[#1]}
\newcommand{\ActLozenge}[1]{\langle #1\rangle}
\newcommand{\knowsLozenge}[1]{\widehat{\mathsf{K}}_{#1}}

\newtheorem{lemma}{Lemma}
\newtheorem{definition}{Definition}
\newtheorem{convention}{Convention}
\newtheorem{example}{Example}
\theoremstyle{definition}
\theoremstyle{remark}
\newtheorem{remark}{Remark}

\title{Reasoning About Probabilities, Actions, and Knowledge in Fuzzy Modal Logic}

\author{%
Daniil Kozhemiachenko$^1$ \and Igor Sedlár$^2$\\
\affiliations
$^1$Aix Marseille Univ, CNRS, LIS, Marseille, France\\
$^2$Institute of Computer Science, The Czech Academy of Sciences, Prague, Czech Republic
\emails
daniil.kozhemiachenko@lis-lab.fr, sedlar@cs.cas.cz}

\begin{document}
\maketitle
\begin{abstract}
We explore a~fuzzy modal logic that can formalise probabilistic reasoning about actions and knowledge. In particular, we deal with contexts involving statements about events expressed via modal formulas, e.g., ‘after doing~$a$, the probability of $A$~knowing that $p$ holds increases~/ decreases / is equal to~$0.25$’, ‘according to~$A$, $p$~is equally likely to happen after doing $a$~or~$b$’, etc. We define the semantics of the logic on Kripke frames equipped with probability measures. We analyse the complexity of deciding the satisfiability of formulas of our logic over finitely branching models, for the full language and its fragments of varying expressivity. In particular, we identify several fragments of our logic where satisfiability is decidable in polynomial time. %polynomially tractable.
\end{abstract}
\allowdisplaybreaks
\section{Introduction\label{sec:introduction}}
Probabilistic reasoning is a~well-established field in artificial intelligence and knowledge representation and reasoning. It has found numerous applications in logic programming~\cite{NgSubrahmanian1992,Lukasiewicz1998,RiguzziSwift2018}, ontologies~\cite{Lukasiewicz2008,Penaloza2020}, and formal argumentation~\cite{LiOrenNorman2012,FazzingaFlescaParisi2015}. An important approach to the formalisation of probabilistic reasoning is \emph{probabilistic logics}. These are formal systems that formalise reasoning with probabilistic statements such as ‘if the probability of~$p$ is twice as much as the probability of~$q$, then the probability of~$q$ is less than~$0.5$’.

In this paper, we consider a~fuzzy modal logic that can formalise probabilistic reasoning about actions and knowledge. We study the complexity of the full language and its fragments, including those for which satisfiability is decidable in polynomial time. This presents a~perspective on the computational properties of formal verification and automated reasoning with probabilistic specifications that describe agents' knowledge and actions.

\paragraph{Probabilistic Logics}
Probabilistic logics were first proposed by Nilsson~\shortcite{Nilsson1986}. The propositional fragment of Nilsson's logic was further studied by Fagin et al.~\shortcite{FaginHalpernMegiddo1990}. Later, Fagin and Halpern~\shortcite{FaginHalpern1994} considered an expansion of probabilistic logic with epistemic modalities. Bacchus et al.~\shortcite{BacchusEtAl1999} propose a framework for a probabilistic logic of action and belief based on Reiter's situation calculus. Halpern and Pucella~\shortcite{HalpernPucella2002} considered logics with operators expressing upper probabilities to model qualitative uncertainty about probabilities. Halpern and O'Neill~\shortcite{HalpernO’Neill2005} apply probabilistic epistemic logic to reasoning about anonymity in multi-agent systems.

Since then, even more expressive systems of probabilistic logics have been extensively studied. In particular, Ognjanovic et al.~\shortcite{OgnjanovicMarkovicRaskovicDoderPerovic2012}, Doder and Perovic~\shortcite{DoderPerovic2020}, and Doder and Ognjanovic~\shortcite{DoderOgnjanovic2024} consider temporal probabilistic logic. Tomovic et al.~\shortcite{TomovicOgnjanovicDoder2020} studied an epistemic expansion of the \emph{first-order} probabilistic logic, and Doder et al.~\shortcite{DoderSavicOgnjanovic2020} explored a~multi-agent logic for reasoning about lower and upper probabilities.

\paragraph{Probabilistic Reasoning in Fuzzy Logics}
An important alternative approach to probabilistic logics is to use so-called \emph{fuzzy logics} to reason about probabilities. Fuzzy logics were introduced by Zadeh~(\citeyear{Zadeh1965,Zadeh1975}) to reason about imprecise statements such as ‘it is cold outside’, ‘the symptoms are severe’, etc. In fuzzy logics, formulas are evaluated over the interval $[0,1]$, and their values are interpreted as \emph{degrees of truth} from $0$~(absolutely false) to~$1$ (absolutely true). Fuzzy logics have multiple applications in artificial intelligence. In particular, He et al.~\shortcite{HeLeungJennings2003} develop an autonomous bidding strategy based on fuzzy logic. Semantics of fuzzy logic is used to support learning and reasoning in real-world domains~\cite{DiligentiGoriSacca2017,BadreddinedAvila-GarcezSerafiniSpranger2022} and statistical relational learning~\cite{BachBroechlerHuangGetoor2017}. Furthermore, van Krieken et al.~\shortcite{vanKriekenAcarvanHarmelen2022} analyse the interaction between different fuzzy logic semantics and learning. Chen et al.~\shortcite{ChenHuSun2022} use fuzzy logic in reasoning problems with knowledge graphs.%, and Hanikov\'{a} et al.~\shortcite{HanikovaManyaVidal2023} apply it to the MaxSAT problem for MV-algebras.

Reasoning about uncertainty in fuzzy logics is usually formalised with \emph{two-layered formulas}. Namely, the probability of an event $\alpha$ (\emph{inner} or \emph{event formula}) is the truth degree of the \emph{probabilistic atom}~--- the formula $\Prob(\alpha)$, interpreted as ‘$\alpha$ is probable’. Initially~\cite{HajekHarmancova1995,HajekEtAl1995}, the events were formalised using classical propositional formulas, and probabilistic atoms were combined with connectives of {\L}ukasiewicz logic into \emph{outer formulas}. H\'{a}jek et al.~\shortcite{HajekEtAl2000} expand this framework into the language of Product fuzzy logic with rational constants~\cite{EstevaEtAl2001} to express thresholds and conditional probabilities. It is shown by Baldi et al.~\shortcite{BaldiCintulaNoguera2020} that propositional fuzzy probabilistic logics have the same expressivity as logics of Fagin et al.~\shortcite{FaginHalpernMegiddo1990}. It is argued that the same expressive power comes with a much simpler syntax and axiomatisation.

Modal expansions of fuzzy probabilistic logics have also been extensively studied. In particular, Godo et al.~\shortcite{GodoHajekEsteva2003} expressed Dempster--Shafer belief functions with formulas of the form $\Prob(\Box\alpha)$. Using the work of Marchioni~\shortcite{Marchioni2008}, Corsi et al.~\shortcite{CorsiFlaminioGodoHosni2023} showed that lower probabilities can be expressed via formulas of the form $\Box\Prob(\alpha)$. They considered a~logic where epistemic modalities are used to both express events and statements about their probabilities. Majer and Sedl\'{a}r~\shortcite{MajerSedlar2025} and Kozhemiachenko and Sedl\'{a}r~\shortcite{KozhemiachenkoSedlar2025} considered a~fuzzy probabilistic modal logic with arbitrary $\Box$-like modalities and propositional event formulas, allowing for a broader model of lower probabilities where uncertainty is not necessarily epistemic in nature, but may arise as a result of non-determinism in outcomes of actions.

\paragraph{Contributions}
Although probabilistic fuzzy modal logics have received much attention, to the best of our knowledge, the framework proposed by Corsi et al.~\shortcite{CorsiFlaminioGodoHosni2023} is the only one where modalities can appear on both layers. It also seems that the computational properties of such logics are considerably less studied. Thus, in this paper, we bring together the frameworks of~\cite{CorsiFlaminioGodoHosni2023} and~\cite{MajerSedlar2025} and explore a~logic that we call $\prea$, whose language is based upon a~modal expansion of the combined Łukasiewicz and Product logics where epistemic and action modalities appear both in inner- and outer-layer formulas. As we will see in Section~\ref{sec:expressivity}, this enables us to reason with a~wide range of events and express various relations between their probabilities.

We show that the language can (i)~express probabilistic statements about modal events involving knowledge or effects of actions (e.g.~statements about the probability of an agent knowing a proposition or~a proposition holding after a~specific action is performed) and (ii)~represent lower and upper probabilities of events arising from epistemic uncertainty or from uncertainty about outcomes of actions. As we discuss in more detail below, system specifications of these forms are important in areas such as multi-agent systems, cybersecurity, and robotics.

The main technical result in the paper is an extensive study of the complexity of $\prea$. We establish that the validity (and satisfiability) problem for arbitrary formulas over finitely-branching models is $\pspace$-complete. We also consider several fragments of our language where \emph{validity is decidable in polynomial time}. Our contributions shed light on the computational features of formal verification and automated reasoning for probabilistic action-epistemic specifications within our expressive framework.

\paragraph{Plan of the Paper}
The text is structured as follows. In Section~\ref{sec:language}, we present the language and semantics of $\prea$. We discuss its expressivity in Section~\ref{sec:expressivity}. In Sections~\ref{sec:decidability} and~\ref{sec:polynomial}, we study the complexity of $\prea$ and its fragments. In Section~\ref{sec:relatedwork}, we discuss related work in modal probabilistic logics. We summarise our results and outline the future work in Section~\ref{sec:conclusion}. Omitted proofs are put in the appendix of the technical report~(\cite{KozhemiachenkoSedlar2026arxiv}).
\section{Language and Semantics\label{sec:language}}
We begin with the definition of the \emph{event language}.
\begin{definition}\label{def:Lev}
Let $\Ac$ and $\Ag$ be two countable sets of labels, representing the set of basic actions and the set of agents, respectively. Let $\Var$ be a countable set of propositional variables, expressing basic events. The set $\Lev$ of \emph{event formulas} is defined using the following grammar:
\begin{align*}
\Lev\ni\alpha& \coloneqq p \mid \sneg \alpha \mid \alpha \land \alpha \mid \knows_A \alpha \mid \ActBox{a}\alpha
\end{align*} 
where $p \in \Var$, $A \in \Ag$ and $a \in \Ac$.
\end{definition}
As one can see, $\Lev$ is a~variant of the classical multi-modal language containing knowledge modalities $\knows_A$ for agents $A \in \Ag$ (where $\knows_A \alpha$ means that `the agent $A$ knows that $\alpha$') in combination with action modalities $\ActBox{a}$ for actions $a \in \Ac$ (where $\ActBox{a}\alpha$ means that `$\alpha$ holds in all states that can be obtained by doing $a$'). %possible
\begin{definition}[Event frames and models]~
\begin{itemize}
\item An \emph{event frame} is a~tuple $\Ffrak=\langle W,\Emsf,\Rmsf\rangle$ where $W\neq\varnothing$, $\Emsf$~is a~function from $\Ag$ to the set of equivalence relations on~$W$, and $\Rmsf:\Ac\rightarrow2^{W\times W}$.
\item An \emph{event model} is a tuple $\Efrak=\langle W,\Emsf,\Rmsf,V\rangle$ with $\langle W,\Emsf,\Rmsf\rangle$ being an event frame and $V:\Var\rightarrow2^W$.
\end{itemize}
\end{definition}
Given a~model~$\Efrak$, we will write $W_\Efrak$, $\Emsf_\Efrak$, $\Rmsf_\Efrak$, and $V_\Efrak$ to denote, respectively, the set of states of~$\Efrak$, its epistemic and action relations, and its valuation.

Informally, $W$ is a set of possible states of the environment (or possible worlds). The equivalence relation $\Emsf_A$ for $A \in \Ag$ represents the epistemic state of agent $A$: $\langle w, u \rangle \in \Emsf_A$ (written also as $\Emsf_A wu$ or $w\Emsf_Au$) means that, in state $w$, the agent $A$ lacks sufficient information to distinguish state $w$ from state $u$. Similarly, the relation $\Rmsf_a$ represents possible outcomes of performing action $a \in \Ac$: $\langle w, u \rangle \in \Rmsf_a$ (written also as $\Rmsf_a wu$ or $w\Rmsf_au$) means that $u$~is a~possible outcome of performing~$a$ in state~$w$. That is, we do not assume that actions are deterministic. We define $\Rmsf_a(w) \coloneqq \{w'\mid w\Rmsf_aw'\}$ and $\Emsf_A(w)\coloneq\{w'\mid w\Emsf_Aw'\}$. Finally, the valuation function~$V$ assigns to each propositional variable $p \in \Var$ an \emph{event}, that is, a~subset of~$W$. This event will also be denoted as~$\| p \|_\Efrak$. Events $\| \alpha \|_\Efrak$ for complex event formulas are defined as usual in modal logic:
\begin{align*}
\| \sneg \alpha \|_\Efrak&=W \setminus \| \alpha \|_\Efrak\\
\| \alpha_1 \land \alpha_2 \|_\Efrak&= \| \alpha_1 \|_\Efrak \cap \| \alpha_2 \|_\Efrak\\
\| \knows_A \alpha \|_\Efrak&= \Kmc_A\| \alpha \|_\Efrak\coloneqq \{ w \mid  \Emsf_A(w) \subseteq \| \alpha \|_\Efrak\}\\
\|\ActBox{a}\alpha \|_\Efrak&= \mathcal{R}_a \| \alpha \|_\Efrak \coloneqq \{ w \mid \Rmsf_a(w) \subseteq \| \alpha \|_\Efrak \}
\end{align*}
\begin{definition}[Probability space]\label{def:probabilityspace}
A \emph{probability space} over an event model $\Efrak{=}\langle W,\Emsf,\Rmsf,V\rangle$ is a pair $\Smbb{=}\langle\Fmc,\mu\rangle$ s.t.:
\begin{itemize}[noitemsep,topsep=0pt]
\item $\Fmc$ is a collection of subsets of $W$ which contains $W$ and $V(p)$ for all $p \in \Var$; and is closed under complements, finite unions and the operators $\Kmc_A$, $\mathcal{R}_a$ for all $A \in\Ag$ and $a \in \Ac$;
\item $\mu$ is a finitely additive probability measure on $\Fmc$.
\end{itemize}    
\end{definition}

\begin{definition}[Probabilistic models]\label{def:probabilitsticmodel}
A \emph{probabilistic model} is a tuple $\Ffrak=\langle W,\Fmbf,\Emsf,\Rmsf,\mu,V\rangle$ s.t.\ $\langle W,\Emsf,\Rmsf,V\rangle$ is an event model and $\langle\Fmbf,\mu\rangle=\langle\Fmc_{w,A},\mu_{w,A}\rangle_{w\in W,A\in\Ag}$ is a~tuple of probability spaces over $\langle W,\Emsf,\Rmsf,V\rangle$ for each $w \in W$ and $A\in\Ag$.
\end{definition}
%%%%%

Intuitively, $\mu_{w,A}(X)$ is the subjective probability of~$X$ for agent~$A$ in state~$w$. To achieve generality, we do not assume that the same events are measurable for each agent in each state. This design choice gives rise to the plurality of probability spaces $\langle \Fmbf, \mu\rangle$. To express \emph{reasoning} about events and their probabilities, we will use a~language based on the fuzzy modal logic. We thus undertake the approach proposed by~Hájek et al.~(\citeyear{HajekEtAl1995,Hajek1998,HajekEtAl2000}). The propositional fragment of our language is based on that of $\Luk\bm{\Pi}\onehalf$~--- the logic that combines connectives of Łukasiewicz and Product logics. We further expand it with epistemic and action modalities corresponding to those in the event language.
\begin{definition}\label{def:Lpr}
The \emph{full probabilistic language} $\Lpr$ is defined using the following grammar:
\begin{align*}
\phi&\coloneqq \Prob_A(\alpha) \mid \neg \phi \mid \phi {\to} \phi \mid \phi {\bullet} \phi \mid \phi {\to_\Pi} \phi \mid \onehalf\mid \knows_A \phi \mid \ActBox{a} \phi
\end{align*}
where $\alpha \in \Lev$, $A \in \Ag$ and $a \in \Ac$.
\end{definition}

\begin{definition}\label{def:preasemantics}
Given a probabilistic model $\Mfrak$, we define \emph{$\Mfrak$-interpretation} as a~function $\Imc_\Mfrak:\Lpr\times W_\Mfrak\to [0,1]$ s.t.:
\begin{align*}
\Imc_\Mfrak(\Prob_A(\alpha),w)&= \mu_{w,A} (\| \alpha \|_\Mfrak)\\
\Imc_\Mfrak(\neg\phi,w)&= 1 - \Imc_\Mfrak(\phi,w)\\
\Imc_\Mfrak(\phi{\rightarrow}\chi,w)&=\min(1,1-\Imc_\Mfrak(\phi,w)+\Imc_\Mfrak(\chi,w))\\
\Imc_\Mfrak(\phi\bullet\chi,w)&=\Imc_\Mfrak(\phi,w)\cdot\Imc_\Mfrak(\chi,w)\\
\Imc_\Mfrak(\phi{\rightarrow_\Pi}\chi,w)&=\begin{cases}1\text{ if }\Imc_\Mfrak(\phi,w)\leq\Imc_\Mfrak(\chi,w)\\\dfrac{\Imc_\Mfrak(\chi,w)}{\Imc_\Mfrak(\phi,w)}\text{ otherwise;}\end{cases}\\
\Imc_\Mfrak(\onehalf,w)&= \tfrac{1}{2}\\
\Imc_\Mfrak(\knows_A \phi,w)&= \inf \{ \Imc_\Mfrak(\phi,w') \mid w \Emsf_A w' \}\\
\Imc_\Mfrak(\ActBox{a}\phi,w)&= \inf \{ \Imc_\Mfrak(\phi,w') \mid w \Rmsf_a w' \}
\end{align*}
\end{definition}
\begin{definition}\label{def:satisfiability}
Let $\Gamma\cup\{\chi\}\subseteq\Lpr$ be finite. We say that
\begin{itemize}[noitemsep,topsep=0pt]
\item $\chi$ is \emph{$\prea$-satisfiable} (on a~frame~$\Ffrak$) iff there is a~probabilistic model~$\Mfrak$ (on~$\Ffrak$) and~$w\in\Mfrak$ s.t.\ $\Imc_\Mfrak(\chi,w)=1$;
\item $\chi$ is \emph{$\prea$-valid} (on~$\Ffrak$) iff for every probabilistic model~$\Mfrak$ (on~$\Ffrak$) and every $w\in\Mfrak$, it holds that $\Imc_\Mfrak(\chi,w)=1$;
\item $\Gamma$ \emph{entails}~$\chi$ iff $\Imc_\Mfrak(\chi,w)=1$ for every probabilistic model~$\Mfrak$ and $w\in\Mfrak$ s.t.\ $\Imc_\Mfrak(\phi,w)=1$ for all $\phi\in\Gamma$.
\end{itemize}
\end{definition}
\begin{convention}\label{conv:boxesdiamonds}
We will further use $\ActLozenge{a}\phi$ and $\knowsLozenge{A}\phi$ as shorthands for $\neg\ActBox{a}\neg\phi$ and $\neg\knows_A\neg\phi$, respectively. Modalities $\knows_A$ and $\ActBox{a}$ are called \emph{box-like}, and modalities $\knowsLozenge{A}$ and $\ActLozenge{a}$ \emph{diamond-like}. We will also use $\blacksquare$ to denote (sequences of) arbitrary box-like modalities: $\ActBox{a}\knows_B\ActBox{c}$, $\knows{A}\knows{B}\ActBox{c}$, etc. Similarly, $\blacklozenge$~stands for (sequences) of diamond-like modalities. In addition, we define
\begin{align*}
\one&\coloneqq\onehalf{\rightarrow}\onehalf&\phi{\oplus}\chi&\coloneqq\neg\phi{\rightarrow}\chi&\triangle\phi&\coloneqq\neg(\phi{\rightarrow_\Pi}\zero)\\
\zero&\coloneqq\neg\one&\phi{\odot}\chi&\coloneq\neg(\phi{\rightarrow}\neg\chi)
\end{align*}
The semantics of these connectives is derived from Definition~\ref{def:preasemantics}:
\begin{align*}
\Imc_\Mfrak(\knowsLozenge{A}\phi,w)&=\sup\{\Imc_\Mfrak(\phi,w')\mid w\Emsf_Aw'\}\\
\Imc_\Mfrak(\ActLozenge{a}\phi,w)&=\sup\{\Imc_\Mfrak(\phi,w')\mid w\Rmsf_aw'\}\\
\Imc_\Mfrak(\phi\oplus\chi,w)&=\min(1,\Imc_\Mfrak(\phi,w)+\Imc_\Mfrak(\chi,w))\\
\Imc_\Mfrak(\phi\odot\chi,w)&=\max(0,\Imc_\Mfrak(\phi,w)+\Imc_\Mfrak(\chi,w)-1)
\end{align*}
\begin{align*}
\begin{matrix}\Imc_\Mfrak(\one,w)=1\\\Imc_\Mfrak(\zero,w)=0\end{matrix}&&\Imc_\Mfrak(\triangle\phi,w)=
\begin{cases}
1&\text{if }\Imc_\Mfrak(\phi,w)=1\\0&\text{otherwise}
\end{cases}
\end{align*}
\end{convention}

Let us briefly explain the semantics of connectives in Definition~\ref{def:preasemantics} and Convention~\ref{conv:boxesdiamonds}. One can see that both implications, $\rightarrow$ and $\rightarrow_\Pi$ are order-preserving. $\oplus$~is the truncated sum on~$[0,1]$, and $\triangle$~helps form crisp statements. Moreover, $\neg$, $\odot$, $\oplus$, and $\rightarrow$ interact in the classical way: $\neg\phi\oplus\chi$ is equivalent to $\phi\rightarrow\chi$, and $\neg(\phi\oplus\chi)$ is equivalent to $\neg\phi\odot\neg\chi$. Notice also that $\blacksquare\Prob_A(\alpha)$ and $\blacklozenge\Prob_A(\alpha)$ express \emph{lower} and \emph{upper probability} of~$\alpha$ in accessible states.
\begin{definition}\label{def:lowerupperprobability}
Let $W\neq\varnothing$ and $\Fmc\subseteq2^W$ be given as in Definition~\ref{def:probabilityspace}. Let further, $\Mmc$ be a~set of probability measures on~$\Fmc$ and $X\in\Fmc$. We define the \emph{lower} and \emph{upper} probabilities of~$X$ ($\mu_*(X)$ and $\mu^*(X)$, respectively) as follows:
\begin{align*}
\mu_*(X){=}\inf\{\mu(X){\mid}\mu{\in}\Mmc\}&&
\mu^*(X){=}\sup\{\mu(X){\mid}\mu{\in}\Mmc\}
\end{align*}
\end{definition}

We will conclude the section with a few brief observations. First, for every rational number~$\tfrac{m}{n}$, one can define an $\Lpr$ formula~$\overline{\tfrac{m}{n}}$ s.t.\ $\Imc_\Mfrak(\overline{\tfrac{m}{n}},w)=\tfrac{m}{n}$ in all probabilistic models, using $\{\neg,\oplus,\bullet,\rightarrow_\Pi,\onehalf\}$ (cf.~\cite[Theorem~7]{EstevaEtAl2001}). Second, note that satisfiability and validity in $\prea$ are reducible to one another in an expected way. Namely, $\phi$ is satisfiable iff $\neg\triangle\phi$ is \emph{not valid}, and, conversely, $\phi$ is valid iff $\neg\triangle\phi$ is \emph{unsatisfiable}. Third, finitary entailment is reducible to validity: $\phi$ entails $\chi$ iff $\triangle\phi\rightarrow\chi$ is valid.
\section{Expressivity\label{sec:expressivity}}

The key feature of our language is its use of modal operators both in the event layer and in the probabilistic layer. Importantly, the modal operators are not limited to $\sfive$-style epistemic operators expressing `agent $A$ knows that\ldots', but they include general $\Kmbf$-style action operators expressing `\ldots holds after action $a$ is performed'. This combined language is rather expressive, as we show in this section.

We first point out that the event layer of our language can express various kinds of \emph{modal events}.
\begin{example}[Some modal events]\label{example:modalevents}~
\begin{itemize}[noitemsep,topsep=0pt]
    \item `After action $a$ is performed, $p$ will hold':\mbox{}\hfill $\ActBox{a}p$.
    \item `$A$ knows that $p$':\mbox{}\hfill $\knows_A p$.
    \item `After~$a$ is performed, $A$~will know that~$p$':\mbox{}\hfill $\ActBox{a}\knows_A p$.
    \item `$A$ knows that after~$a$ is performed, $p$ will hold':\mbox{}\hfill $\knows_A \ActBox{a}p$.
\end{itemize}
\end{example}

Our probabilistic operator $\Prob_A$ that expresses \emph{subjective probability} of agent~$A$ can then be used to represent probabilities of these events. We will treat objective probability as a~particular case of subjective probability. Formally, we will assume that $\mathtt{ob}\in\Ag$ and write $\Prob_\mathtt{ob}$ in our examples when dealing with objective probability.

Reasoning about probabilities of such modal events is crucial in various areas. In multi-agent systems, e.g., card games~\cite{BrownSandholm2019}, agents need to reason about what other agents know, and this reasoning is often based on probabilistic information. In cybersecurity, a security system needs to be able to reason probabilistically about what the attacker might know~\cite{BartheKopfOlmedoZanella-Beguelin2013}. Probabilistic robotics \cite{ThrunEtAl2005} is well-studied. In particular, in human-robot interaction scenarios, human controllers often need to take into account what robots with possibly faulty sensors are likely to know in a~given situation~\cite{AkkaladeviPlaschHofmannPichler2021}.

Using propositional connectives of the fuzzy logic $\Luk\Pi\onehalf$, we can express comparisons between probabilistic statements, including threshold comparisons.
\begin{itemize}[noitemsep,topsep=0pt]
    \item `According to~$A$, it's not much more likely that $p$ is going to happen if action $a$ is performed than if action $b$ is performed':
    \begin{align*}
    \Prob_A (\ActBox{a}p) \to \Prob_A(\ActBox{b}p)
    \end{align*}
    \item `Agent $A$ likely knows that $p$':
    \begin{align*}
    \onehalf \to \Prob_\mathtt{ob}(\knows_A p)
    \end{align*}
    \item `The probability that~$A$ knows that $p$ is at least~$0.7$':
    \begin{align*}
    \triangle(\overline{0.7}\to\Prob_\mathtt{ob}(\knows_A p))
    \end{align*}
    \item `According to agent $A$, it's twice as likely for $p$ to hold after $a$ is performed that it is for $p$ to fail':
    \begin{align*}
    \triangle (\onehalf \bullet \Prob_A(\ActBox{a}p) \leftrightarrow \Prob_A(\ActBox{a}\neg p))
    \end{align*}
\end{itemize}

The following examples illustrate the expressivity of our logic in the setting of automated bidding agents.
\begin{example}\label{example:bidding1}
Consider two competing automated bidding agents $A$ and $B$. Agent $A$ needs to decide whether to place a~high bid. The decision depends on $A$'s subjective probability of $B$'s knowing that the item they bid for is rare.

The subjective probability is represented by the atomic probabilistic formula $\Prob_A (\knows_B \mathtt{rare})$. The truth degree of the implicational formula $\overline{0.7} \to \Prob_A (\knows_B \mathtt{rare})$ is~$1$ if $A$'s subjective probability of $B$ knowing $\mathtt{rare}$ is no less that $0.7$ and is the difference between $A$'s subjective probability of $B$~knowing $\mathtt{rare}$ and $0.7$ otherwise (cf.\ Definition \ref{def:preasemantics}).

Using the delta operator, we can express the crisp statement that $A$'s subjective probability of $B$ knowing $\mathtt{rare}$ is at least~$0.7$. This is expressed by $\triangle (\overline{0.7} \to \Prob_A (\knows_B \mathtt{rare}))$. Compare this to the imprecise statement that $A$'s subjective probability of $B$ knowing $\mathtt{rare}$ is not \emph{much} less than $0.7$: $\overline{0.7}\to \Prob_A (\knows_B \mathtt{rare})$.
\end{example}

\begin{example}\label{example:bidding2}
Assume that there is an auction with two participants: $A$~and~$B$. One of the things $A$~can do is to place an \emph{aggressive bid} ($a$). Now $A$'s subjective probability that $B$ fill fold after $A$ places an aggressive bid is formalised by the formula $\Prob_A (\ActBox{a} (B\, \mathtt{folds}))$. Using rational constants, we can express bounds on subjective probabilities. E.g., the truth degree of formula $\overline{0.7} \to \Prob_A (\ActBox{a} (B\, \mathtt{folds}))$ is inversely proportional to the amount to which $0.7$ exceeds $A$'s subjective probability that $B$ fill fold after $A$ places an aggressive bid, and $ \triangle (\overline{0.7} \to \Prob_A (\ActBox{a} (B\, \mathtt{folds})))$ says that $A$'s subjective probability that $B$ fill fold after $A$ places an aggressive bid is not less than $0.7$.

Moreover, we can express comparisons between subjective probabilities of the outcomes of different actions. Let $m$ represent agent $A$ placing a~\emph{modest bid}; then $\Prob_A (\ActBox{m} (B \, \mathtt{folds})) \to \Prob_A (\ActBox{a} (B \, \mathtt{folds}))$ formalises the imprecise statement that $A$'s subjective probability that $B$~will fold after $A$~places an aggressive bid is not much less than $A$'s subjective probability that $B$~will fold after $A$~places a~modest bid. Note that the delta operator $\triangle$ can turn an imprecise comparison into a precise one~--- the formula $\triangle (\Prob_A (\ActBox{m} (B \, \mathtt{folds})) \to \Prob_A (\ActBox{a} (B \, \mathtt{folds})))$ says that $A$'s subjective probability that $B$ fill fold after $A$ places an aggressive bid is not less than $A$'s subjective probability that $B$ fill fold after $A$ places a modest bid.   
\end{example}

Modal operators in the probabilistic layer can, in turn, formalise upper and lower probabilities (recall Definition~\ref{def:lowerupperprobability}) expressing qualitative uncertainty about probability. This uncertainty can arise from the lack of information on an agent's side, or from the fact that outcomes of actions are not determined, for example. Upper and lower probabilities are also useful when specifying the effects of actions on probabilities of events.
\begin{itemize}[noitemsep,topsep=0pt]
    \item `Performing $b$ cannot substantially decrease the probability that performing $a$ will lead to $p$ being true':
    \begin{align*}
    \Prob_\mathtt{ob}(\ActBox{a}p) \to\ActBox{b} \Prob_\mathtt{ob}(\ActBox{a}p)
    \end{align*}
    \item `Performing $a$ cannot substantially increase the probability of agent $A$ knowing that $p$':
    \begin{align*}
    \ActBox{a}\Prob_\mathtt{ob} (\knows_A p) \to \Prob_\mathtt{ob}(\knows_A p)
    \end{align*}
    \item `The amount of qualitative uncertainty of agent $A$ concerning event $p$ is not much higher than $0.2$':
    \begin{align*}
    \overline{0.2} \to (\knows_A \Prob_\mathtt{ob}(p)\leftrightarrow \knowsLozenge{A}\Prob_\mathtt{ob}(p))
    \end{align*}
\end{itemize}

\begin{example}\label{example:bidding3}
    In the context of the previous bidding example, assume that $A$ made a subtle, revealing bid. Agent $A$ does not know what $B$'s subjective probability of `$A$ has a strong hand' ($A\,\mathtt{strong}$) is after the revealing bid, $A$ only has a~range of possibilities. The formula $\knows_A (\Prob_B (A\,\mathtt{strong}))$ expresses the infimum of the values of $\Prob_B (A\,\mathtt{strong})$ across these possibilities~--- $B$'s lower subjective probability of $A\,\mathtt{strong}$ across the possibilities left open by agent $A$.
    
    Using {\L}ukasiewicz implication, we can express comparisons between the lower probability and the `actual' probability, for instance. This is done by the formula $\Prob_B (A\,\mathtt{strong}) \to \knows_A (\Prob_B (A\,\mathtt{strong}))$ that expresses the truth degree of the imprecise statement that $B$'s subjective probability of $A\,\mathtt{strong}$ is not much lower than lower $B$'s probability of $A\,\mathtt{strong}$ (i.e.\ `what $A$ considers possible').
    
    We can explicitly express the fact that we are looking at probabilities \emph{after the revealing bid} using the action modality~$\ActBox{r}$. In particular, the formula $\ActBox{r} \Prob_B (A\,\mathtt{strong})$ expresses the infimum of $B$'s subjective probabilities of $A\,\mathtt{strong}$ in all possible outcomes of the action $r$. Similarly, the formula $\ActBox{r} \knows_A (\Prob_B (A\,\mathtt{strong}))$ expresses the infimum of $B$'s lower subjective probability of $A\,\mathtt{strong}$ across the possibilities left open by agent $A$ after $r$. As before, we can use {\L}ukasiewicz implication to express a comparison of these infima: $\ActBox{r} \Prob_B (A\,\mathtt{strong}) \to \ActBox{r}\knows_A (\Prob_B (A\,\mathtt{strong}))$.
\end{example}

We finish the section with the following observation. It is known from~\cite{MajerSedlar2025} that the logic $\KLukProdProb$ proposed there lacks the finite model property (FMP). In particular, $\phi=\triangle\ActLozenge{a}\Prob_A(p)\rightarrow\ActLozenge{a}\triangle\Prob_A(p)$ is valid in all finite frames but can be falsified in a~model~$\Mfrak$ that contains a~state~$w$ s.t.\ $\Rmsf_a(w)=\{w_i\mid i\in\Nmbb\}$; it suffices to set $\Imc_\Mfrak(\Prob_A(p),w_i)=\tfrac{i}{i+1}$. Dually, $\neg\triangle\phi$ is unsatisfiable on every finite frame, but can be satisfied on the same infinitely branching model that falsifies~$\phi$. As $\KLukProdProb$ is a~fragment of~$\prea$, it follows that the latter also lacks FMP.

In most practical applications, however, it is sufficient to consider finite models. Still, the technical results presented in the next section apply to a~slightly more general class of models: the \emph{finitely branching} ones. These models represent situations in which an action can have only a finite number of possible outcomes, or in which an agent considers only a finite number of possibilities (i.e.\ their information in any state is consistent with only a finite number of states).

A~similar restriction on models is often used in fuzzy description logics. Namely (cf.~work by H\'{a}jek~\shortcite{Hajek2007}, Bobillo et al.~\shortcite{BobilloDelgadoGomez-RamiroStraccia2009,BobilloDelgadoGomez-RamiroStraccia2012} and Borgwardt et al.~\shortcite{BorgwardtDistelPenaloza2014KR,BorgwardtDistelPenaloza2014DL,BorgwardtCeramiPenaloza2015,BorgwardtCeramiPenaloza2017}), one may consider \emph{witnessed} interpretations, i.e., those with the following property. If a~quantified assertion $[\forall\Rmsf.\Cmsf](a)$ has value~$x$, then there must be an individual~$b$ s.t.\ $\Rmsf(a,b)$ and $\Cmsf(b)$ has value~$x$.

In the remainder of the paper, we will mostly consider the satisfiability of $\Lpr$~formulas over \emph{finitely branching} frames. In what follows, we will use $\preafb$ to stand for the logic $\prea$ interpreted over finitely-branching frames. We will also say that $\phi\in\Lpr$ is \emph{$\preafb$-valid} (\emph{$\preafb$}-satisfiable) meaning that $\phi$~is $\prea$-valid on all ($\prea$-satisfiable on some) finitely branching frames.
\section{Decidability of~$\preafb$\label{sec:decidability}}
In this section, we will establish the decidability of $\preafb$-satisfiability (and hence, validity). To do that, we will proceed as follows. We begin with constructing a~terminating tableaux calculus for $\EALukProdfb$~--- the fragment of~$\preafb$ \emph{without event formulas}. Then we will use these tableaux to devise a~decision procedure for $\preafb$.

Let us first define $\Lmod$~--- the language of~$\EALukProdfb$.
\begin{definition}\label{def:Lmod}
Let $p \in \Var$, $A \in \Ag$, and $a \in \Ac$, we define $\phi\in\Lmod$ as follows:
\begin{align*}
\phi&\coloneqq p\mid\neg\phi\mid\phi\rightarrow\phi\mid\phi\bullet\phi\mid\phi\rightarrow_\Pi\phi\mid\onehalf\mid\knows_A\phi\mid\ActBox{a}\phi
\end{align*}
The semantics of~$\EALukProdfb$ can be obtained in the expected manner by restricting the interpretation functions from Definition~\ref{def:preasemantics} onto \emph{finitely branching event frames}.
\end{definition}
\begin{definition}\label{def:EALukProdsemantics}
An~$\EALukProdfb$-model is $\Mfrak=\langle W,\Emsf,\Rmsf,v\rangle$ with $\langle W,\Emsf,\Rmsf\rangle$ being a~finitely branching event frame and $v\colon\Var\rightarrow[0,1]$. The  \emph{$\EALukProdfb$-interpretation induced by~$\Mfrak$} is $\Imc_\Mfrak\colon W\times\Lmod\rightarrow[0,1]$ s.t.\ $\Imc_\Mfrak(p,w)=v(p,w)$, and the values of complex formulas are computed according to Definition~\ref{def:preasemantics}. The notions of $\EALukProdfb$-sa\-tis\-fi\-a\-bi\-li\-ty, validity, and entailment are also adapted from Definition~\ref{def:preasemantics}.
\end{definition}
\begin{remark}\label{rem:fewermodels}
We note that $\EALukProd$ and $\prea$ are connected in the following way. Given $\phi{\in}\Lpr$, one can think of it as an $\Lmod$-formula $\phi^\uparrow$ where each probabilistic atom $\Prob_A(\alpha)$ is replaced by a~fresh variable~$p_{A,\alpha}$. In this interpretation, $\phi$~is satisfied in a~probabilistic model~$\Mfrak=\langle W,\Fmbf,\Emsf,\Rmsf,\mu,V\rangle$ iff $\phi^\uparrow$ is satisfied in the \emph{$\EALukProd$-model} $\Mfrak^\uparrow=\langle W,\Emsf,\Rmsf,V^\uparrow\rangle$ on the same frame s.t.\ $V^\uparrow(p_{A,\alpha},w)=\Imc_\Mfrak(\Prob_A(\alpha))$.  
\end{remark}

\begin{convention}\label{conv:formulalength}
Given a~formula~$\phi$ and a~set of formulas~$\Gamma$, we use $\lmc(\phi)$ and~$\lmc[\Gamma]$ to denote their \emph{lengths}, i.e., the number of symbols in them.
\end{convention}

Let us now present a~tableaux calculus for~$\EALukProdfb$. We adapt the idea of ‘constraint tableaux’ by H\"{a}hnle~\shortcite{Haehnle1999} for the modal logic framework. We first define \emph{constraints} that encode the values of formulas in states and \emph{relational terms} that represent relations between different states on the branch of the tableau.
\begin{definition}[Constraints]\label{def:tableauxlanguage}
We fix a~countable set $\Wmsf=\{w,w',w_0,w_1,\ldots\}$ of \emph{state labels} and define:
\begin{itemize}[noitemsep,topsep=0pt]
\item a~\emph{labelled formula} as $w{:}\phi$ with $w\in\Wmsf$ and $\phi\in\Lmod$;
\item a~\emph{formulaic constraint} as $w{:}\phi\triangledown P$ s.t.\ $w{:}\phi$ is a~labelled formula, $\triangledown\in\{\leq,\geq\}$ and $P$~a~polynomial over~$[0,1]$ with integer coefficients;
\item a~\emph{numerical constraint} as $w{:}P\triangleleft P'$ with $w{\in}\Wmsf$, $\triangleleft{\in}\{\leq,<\}$, and~$P$ and~$P'$ polynomials.
\end{itemize}
\end{definition}
\begin{definition}[Relational terms]\label{def:relationalterms}
For sets $\Ac$ and $\Ag$ of \emph{actions} and \emph{agents}, we fix sets of \emph{action} and \emph{epistemic relational labels}~--- $\ActionSet=\{\Rmsf_a\mid a\in\Ac\}$ and $\ClusterSet=\{\cluster^A_i\mid A\in\Ag,i\in\Nmbb\}$ --- and define:
\begin{itemize}[noitemsep,topsep=0pt]
\item an \emph{action relational term} to be $w\Rmsf_aw'$ with $w,w'\in\Wmsf$ and $\Rmsf_a\in\ActionSet$;
\item an \emph{epistemic relational term} to be $w\in\cluster^A_i$ with $w\in\Wmsf$, $i\in\Nmbb$, and $\cluster^A_i\in\ClusterSet$.
\end{itemize}
\end{definition}

Note from the definition above that epistemic labels correspond to \emph{subsets} of the model. Namely, given a~model~$\Mfrak$ each~$\cluster^A_i$ is an equivalence class on~$\Mfrak$ under~$E_A$.

We are now ready to present the rules and formally define the notion of a~\emph{tableaux proof}. Our calculus $\TEALukProd$ expands~$\TKLukProd$ by Kozhemiachenko and Sedl\'{a}r~\shortcite{KozhemiachenkoSedlar2025} with rules for epistemic modalities.
\begin{figure*}[t]
\centering
\begin{align*}
\knows_\leq{:}\dfrac{\begin{matrix}w{:}\knows_A\phi\leq P\\w\in\cluster^A_i\end{matrix}}{\begin{matrix}w'{:}\phi\leq P;~w'\in\cluster^A_i\\w'\in\cluster^{A'}_j(A'\in\Ag)\end{matrix}}
&&
\knows_\geq{:}\dfrac{\begin{matrix}w{:}\knows_A\phi\geq P\\w\in\cluster^A_i~u\in\cluster^A_i\end{matrix}}{u{:}\phi\geq P}&&\Box_\leq{:}\dfrac{w{:}\ActBox{a}\phi\leq P}{\begin{matrix}w\Rmsf_aw';~w'{:}\phi\leq P\\w'\in\cluster^{A'}_j(A'\in\Ag)\end{matrix}}&&\Box_\geq{:}\dfrac{\begin{matrix}w{:}\ActBox{a}\phi\geq P\\w\Rmsf_au\end{matrix}}{u{:}\phi\geq P}
\end{align*}
\begin{align*}
\neg_\leq{:}\dfrac{w{:}\neg\phi\leq P}{w{:}\phi\geq1-P}&&\neg_\geq{:}\dfrac{w{:}\neg\phi\geq P}{w{:}\phi\leq1-P}&&
\rightarrow_\leq{:}\dfrac{w{:}\phi\rightarrow\chi\leq P}{P\geq1\left|\begin{matrix}w{:}\phi\geq1-P+y\\w{:}\chi\leq y\\w{:}y\leq P\end{matrix}\right.}&&\rightarrow_\geq{:}\dfrac{w{:}\phi\rightarrow\chi\geq P}{\begin{matrix}w{:}\phi\leq1-P+y\\w{:}\chi\geq y\end{matrix}}
\end{align*}
\begin{align*}
\bullet{:}\dfrac{w{:}\phi\bullet\chi\triangledown P}{\begin{matrix}w{:}\phi=y_1\\w{:}\chi=y_2\\w{:}y_1\cdot y_2\triangledown P\end{matrix}}
&&
{\rightarrow_\Pi}_\leq{:}\dfrac{w{:}\phi\rightarrow_\Pi\chi\leq P}{P\geq1\left|\begin{matrix}w{:}\phi=y_1\\w{:}\chi=y_2\\w{:}y_2/y_1\leq P\end{matrix}\right.}
&&
{\rightarrow_\Pi}_\geq{:}\dfrac{w{:}\phi\rightarrow_\Pi\chi\geq P}{\begin{matrix}w{:}\phi\geq y\\w{:}\chi\leq y\end{matrix}\left|\begin{matrix}w{:}\phi=y_1\\w{:}\chi=y_2\\w{:}y_2/y_1\geq P\end{matrix}\right.}
\end{align*}
\caption{Tableaux rules: $\triangledown\in\{\leq,\geq\}$; $y$, $y_1$, and $y_2$ are new variables; $w'$ is fresh on the branch, $u$ is present on the branch; $w{:}\psi=y$ is a~shorthand for $\{w{:}\psi\geq y,w{:}\psi\leq y\}$; for each state generated by $\knows_\leq$ and $\Box_\leq$, we add a~fresh relational term $w\in\cluster^{A'}_j$ for every $A'\in\Ag$ and $j\in\Nmbb$ being new on the branch.}
\label{fig:tableauxrules}
\end{figure*}
\begin{definition}[$\TEALukProd$ --- constraint tableaux for $\EALukProd_\fb$]\label{def:KLuktableaux}
A~\emph{constraint tableau} is a~downward branching tree of \emph{formulaic constraints}, \emph{numerical constraints}, and \emph{relational terms}. A~branch can be extended by an application of a~rule from Fig.~\ref{fig:tableauxrules}. Given a~tableau branch~$\Bmc$, we let $\Wmsf_\Bmc=\{w\mid w\in\Wmsf,\,w\text{ occurs on }\Bmc\}$ and define for every $w\in\Wmsf_\Bmc$ the sets of \emph{formulaic}~($\Fmc_w$) and \emph{numerical}~($\Nmc_w$) constraints as well as their corresponding systems of inequalities~$\Fmc^\Imbb_w$ and~$\Nmc^\Imbb_w$ (below, $x_{w{:}\phi}$ is a~variable over $[0,1]$):
\begin{align*}
\Fmc_w&=\{w{:}\phi\triangledown P\mid w{:}\phi\triangledown P\in\Bmc\}\\
\Nmc_w&=\{w{:}P\triangledown P'\mid w{:}P\triangledown P'{\in}\Bmc\}{\cup}\{w{:}\onehalf\triangledown P\mid w{:}\onehalf\triangledown P{\in}\Bmc\}\\
\Fmc^\Imbb_w&=\{x_{w:\phi}\triangledown P\mid w{:}\phi\triangledown P{\in}\Fmc_w\}\\
\Nmc^\Imbb_w&=\{P\triangledown P'\!\mid w{:}(P\triangledown P'){\in}\Nmc_w\}
\end{align*}
Furthermore, we set
\begin{align*}
\Fmc&=\bigcup_{w\in\Wmsf_\Bmc}\Fmc_w&\Nmc&=\bigcup_{w\in\Wmsf_\Bmc}\Nmc_w&\Bmc^\Imbb_w&=\Fmc^\Imbb_w\cup\Nmc^\Imbb_w\\
\Fmc^\Imbb&=\bigcup_{w\in\Wmsf_\Bmc}\Fmc^\Imbb_w&\Nmc^\Imbb&=\bigcup_{w\in\Wmsf_\Bmc}\Nmc^\Imbb_w&\Bmc^\Imbb&=\bigcup_{w\in\Wmsf_\Bmc}\Bmc^\Imbb_w
\end{align*}

A~tableau branch $\Bmc$ is called \emph{open} if its corresponding system of inequalities $\Bmc^\Imbb$ has a~solution over $[0,1]$ and is \emph{closed} otherwise. An open branch is \emph{complete} if for every premise of a~rule present on the branch, a~conclusion of the rule also occurs on the branch. A~tableau is closed when all its branches are closed. A~formula $\phi$ \emph{has a~$\TEALukProd$-proof} if there is a~closed tableau starting with $\{w{:}\phi{\leq}y,w{:}y{<}1\}\cup\{w\in\cluster^A_0\mid A\in\Ag\}$ with $y$~being a~variable.
\end{definition}

Observe briefly that for each generated state~$w'$ on the branch, we add terms $w'\in\cluster^A_i$ where $\cluster^A_i$ is a~\emph{fresh} epistemic relational label and $A\in\Ag$. This works as follows. E.g., we want to apply the $\Box_\leq$ rule to the constraint $w{:}\ActBox{a}\phi\leq P$ on~$\Bmc$. Say that $\Ag=\{A,B\}$, and we already have labels $\cluster^A_1$ and $\cluster^B_1$ on the branch. Then, we extend $\Bmc$ with $w\Rmsf_aw'$ and $w'{:}\phi\leq P$ with $w'$~being a~fresh state label. We also add relational terms $w'\in\cluster^A_2$ and $w'\in\cluster^B_2$. This ensures that the accessibility relations $\Emsf_A$ and $\Emsf_B$ that correspond to labels $\cluster^A_i$ and $\cluster^B_j$ are reflexive. The rule~$\knows_\leq$ works similarly. Note that $\Box_\leq$ and $\knows_\leq$ are sound only w.r.t.\ finitely-branching frames because in general, $\inf\{\Imc_\Mfrak(\phi,w')\mid w\Rmsf_aw'\}\leq x$ does not imply that there is some $w'\in\Rmsf_a(w)$ s.t.\ $\Imc_\Mfrak(\phi,w')\leq x$.

To prove the completeness of the tableaux, we will show that every \emph{complete open} branch~$\Bmc$ has a~model that \emph{realises} constraints on it. The following definition is an application of the standard notion to the semantics of~$\EALukProdfb$.
\begin{definition}[Realising model]\label{def:realisingmodel}
Let~$\Bmc$ be a~tableau branch, $\Bmc^\Imbb$~its corresponding system of inequalities, and $\Smc_{\Bmc^\Imbb}$ a~solution of~$\Bmc^\Imbb$. Given $x_{w{:}\phi}$ occurring in~$\Fmc^\Imbb$, we use $\Smc_{\Bmc^\Imbb}(x_{w{:}\phi})$ to denote its value under~$\Smc_{\Bmc^\Imbb}$. Let further, $\Mfrak=\langle W,\Emsf,\Rmsf,v\rangle$ be an $\EALukProd$-model. $\Mfrak$~\emph{realises~$\Bmc$ under $\Smc_{\Bmc^\Imbb}$} if there is a~map $\real:\Wmsf_\Bmc\rightarrow W$ s.t.\ $\Imc_\Mfrak(\phi,\real(w))=\Smc_{\Bmc^\Imbb}(x_{w{:}\phi})$. $\Mfrak$ is \emph{a~realising model of~$\Bmc$} if there is a~solution $\Smc_{\Bmc^\Imbb}$ s.t.\ $\Mfrak$~realises~$\Bmc$ under $\Smc_{\Bmc^\Imbb}$.
\end{definition}

The next statement can be proved in a~standard manner. For soundness, we show that if~$\Mfrak$ realises a~premise of a~rule, then $\Mfrak$~realises its conclusion. For completeness, we prove that \emph{complete open} branches have realising models.
\begin{restatable}{theorem}{tableauxcompleteness}\label{theorem:tableauxcompleteness}
$\TEALukProd$ is sound and complete w.r.t.~$\EALukProd_\fb$:
\begin{enumerate}[noitemsep,topsep=0pt]
\item if $\phi$ has a~$\TEALukProd$-proof, then it is $\EALukProd_\fb$-valid;
\item if $\phi$ is $\EALukProd_\fb$-valid, then it has a~$\TEALukProd$-proof.
\end{enumerate}
\end{restatable}

Now, let us show that $\preafb$-satisfiability is decidable. For this, we will define an alternative semantics of~$\preafb$ based on \emph{sample-independent models} (SI models) proposed by Kozhemiachenko and Sedlar~\shortcite{KozhemiachenkoSedlar2025}.
\begin{definition}[Sample-independent semantics]\label{def:SImodelsS5events}~
\begin{itemize}[noitemsep,topsep=0pt]
\item An \emph{SI probabilistic frame over $\Xfrak$} (SI-frame) is a~tuple $\Ffrak=\langle W,\Emsf,\Rmsf,\Xfrak,\mu\rangle$ with $\langle W,\Emsf,\Rmsf\rangle$ and $\Xfrak=\langle W_\Xfrak,\Emsf_\Xfrak,\Rmsf_\Xfrak\rangle$ being event frames, and $\mu=\langle\mu_{w,A}\rangle_{w\in W,A\in\Ag}$ being a~tuple of finitely additive probability measures on~$2^{W_\Xfrak}$.
\item An \emph{SI-$\prea$-model} is a~tuple $\Mfrak=\langle W,\Emsf,\Rmsf,\Xfrak,\mu,V\rangle$ with $\langle W,\Emsf,\Rmsf,\Xfrak,\mu\rangle$ being an~SI-frame and $V:\Var\rightarrow2^{W_\Xfrak}$.
\end{itemize}
Given an SI-$\prea$-model $\Mfrak=\langle W,\Emsf,\Rmsf,\Xfrak,\mu,V\rangle$, we call $\Xfrak$ the \emph{inner frame} and $\langle W,\Emsf,\Rmsf\rangle$ the \emph{outer frame}. States of inner and outer frames are called inner and outer states, respectively. The notions of probabilistic interpretations induced by~$\Mfrak$, validity, and satisfiability are adapted from Definition~\ref{def:preasemantics}.
\end{definition}
From the definition above, it is clear that every ‘standard’ probabilistic model (as in Definition~\ref{def:probabilitsticmodel}) can be treated as a~sample-independent model. Thus, if $\phi\in\Lpr$ is valid on all SI-models, it is valid on all ‘standard’ probabilistic models. For the converse direction, we can take an SI-model $\Mfrak=\langle W,\Emsf,\Rmsf,\Xfrak,\mu,V\rangle$ with $\Xfrak=\langle W_\Xfrak,\Emsf_\Xfrak,\Rmsf_\Xfrak,V\rangle$ s.t.\ $\Imc_\Mfrak(\phi,w){<}1$ for some $w{\in}W$. We then define the pro\-ba\-bi\-lis\-tic model $\Mfrak'=\langle W',\Fmbf',\Emsf',\Rmsf',\mu',V'\rangle$:
\begin{itemize}[noitemsep,topsep=0pt]
\item $W'=W\uplus W_\Xfrak$;
\item $\Fmc_{w,A}'=2^{W'}$ for all $w\in W'$ and $A\in\Ag$;
\item $\Emsf'_A=\Emsf_A\uplus \Emsf_{A_\Xfrak}$ for every $A\in\Ag$;
\item $\Rmsf'_a=\Rmsf_a\uplus\Rmsf_{a_\Xfrak}$ for every $a\in\Ac$;
\item $\mu'_{w,A}(\{s\}){=}\mu_{w,A}(\{s\})$ if $s{\in}W_\Xfrak$, $\mu'_{w,A}(\{s\}){=}0$ else;
\item $V'=V$
\end{itemize}
Now, we can show by induction on $\phi\in\Lpr$ that $\Imc_{\Mfrak'}(\phi,w)=\Imc_\Mfrak(\phi,w)$ for every $w\in W$. This gives us the following statement.
\begin{restatable}{proposition}{SIequivalence}\label{prop:SIequivalence}
A~formula $\phi\in\Lpr$ is $\prea$-valid iff it is valid on SI-frames.
\end{restatable}

The idea of the sample-independent semantics is to decouple the event model from the states over which probabilistic formulas are evaluated. Relying on the equivalence between sample-independent semantics and standard semantics, we can prove the $\pspace$-completeness of the satisfiability of $\Lpr$~formulas over finitely branching frames.
\begin{restatable}{theorem}{generalcasedecidability}\label{theorem:generalcasedecidability}
Given $\phi\in\Lpr$, it is $\pspace$-complete to decide whether it is satisfiable over finitely branching frames.
\end{restatable}
\begin{proof}[Proof sketch]
We begin with $\pspace$-hardness. For that, we consider a~polynomial reduction from the satisfiability of $\Lev$-formulas to $\preafb$-satisfiability. Namely, $\alpha\in\Lev$ is satisfiable iff $\Prob_A(\alpha)$ is $\preafb$-satisfiable.

Let us now consider $\pspace$-membership. Our proof combines the ideas of~\cite{FaginHalpern1994} and~\cite{HajekTulipani2001}. First of all, we note that all tableaux terminate: all rules have the branching factor of at most~$2$ and decompose formulas into their subformulas.

Now, we proceed as follows. Using a~standard tableaux procedure akin to the one by Blackburn et al.~\shortcite{BlackburndeRijkeVenema2010}, we begin to construct the frame $\langle W,\Emsf,\Rmsf\rangle$ by building a~tableau for $\{w_0{:}\phi{\geq}1\}$. Note from Remark~\ref{rem:fewermodels} that the tableaux rules are \emph{sound for $\preafb$}. We apply the rules depth-first until we generate a~state~$w$ s.t.\ no labelled formula $w{:}\psi$ contains an outer-layer modality. Then, we decompose formulas labelled with~$w$ and construct a~system of polynomial inequalities $\Bmc^\Imbb_w$ (cf.~Definition~\ref{def:KLuktableaux}).

Now, for each $A\in\Ag$ on~$\Bmc$, we construct two sets
\begin{align*}
\atom_A(w)=&\{\alpha\mid w{:}\Prob_A(\alpha)\triangledown P\in\Bmc\}\\
\subformulas^\sim_A(w)=&\{\beta\mid\exists\alpha\in\atom_A(w):\beta\text{ is a subformula of }\alpha\}\cup\\
&\{\neg\beta\mid\exists\alpha\in\atom_A(w):\beta\text{ is a subformula of }\alpha\}
\end{align*}
We then consider the following event model:
\begin{align*}
\Xfrak_{w,A}&=\langle W^\subformulas,\Emsf^\subformulas,\Rmsf^\subformulas,V^\subformulas\rangle
\end{align*}
The states of~$\Xfrak_{w,A}$ correspond to the maximal consistent subsets of $\subformulas^\sim_A(w)$, the relations are defined as follows:
\begin{align*}
\Rmsf_a^\subformulas(\Xi,\Xi')&\text{ iff }\forall\beta\in\subformulas^\sim_A(w):\ActBox{a}\beta\in \Xi\Rightarrow\beta\in \Xi'\\
\Emsf_A^\subformulas(\Xi,\Xi')&\text{ iff }\forall\beta\in\subformulas^\sim_A(w):\knows_A\beta\in \Xi\Leftrightarrow\knows_A\beta\in \Xi'
\end{align*}
and $V^\subformulas(p)=\{\Xi\mid p\in \Xi\}$ for all $p\in\Var(\phi)$. $\Xfrak_{w,A}$~has exponential size, but for any $\Xi\subseteq\subformulas^\sim_A(w)$, one can check whether $\Xi\in W^\subformulas$ using polynomial space since event formulas contain only $\sfive$ and $\Kmbf$ modalities. Moreover, one can show as in~\cite[Theorem~4.1]{FaginHalpern1994} that for every $\beta\in\subformulas^\sim_A(w)$ and $\Xi\in W^\subformulas$, we have $\Xfrak_{w,A},\Xi\vDash\beta$ iff $\beta\in \Xi$.

We now need to ensure that the values of~$\Prob_A(\alpha)$'s at~$w$ are coherent with a~measure on~$\Xfrak_{w,A}$, i.e., that there is some measure $\mu_{w,A}$ s.t.\ $\mu_{w,A}(\|\alpha\|_{\Xfrak_{w,A}})=\Prob_A(\alpha)$ for all $\alpha\in\atom_A(w)$. We add fresh variables $y^A_{\alpha}$ for every $\alpha\in\atom_A(w)$ and $u^A_\Xi$ for every $\Xi\subseteq\subformulas^\sim_A(w)$. Here, $y$'s are auxiliary variables that will stand for the values of $\Prob_A(\alpha)$'s at~$w$, and $u$'s are the values of $\mu_{w,A}(\{\Xi\})$'s (i.e., we encode the \emph{probability distribution} induced by~$\mu_{w,A}$). Furthermore, we set $a_{\Xi,\alpha}=1$ if $\alpha\in \Xi$ and $a_{\Xi,\alpha}=0$, otherwise. Now, we consider the following system of linear equations:
\begin{align}\label{equ:biglineq}
\left\{\sum_{\Xi\,\subseteq\,\subformulas^\sim_A(w)}u^A_\Xi=1\right\}&\cup\tag{$\mathsf{big}^A_w$}\\
\left\{\sum_{\Xi\,\subseteq\,\subformulas^\sim_A(w)}(a_{\Xi,\alpha}\cdot u^A_\Xi)=y^A_\alpha~\middle|~\alpha\in\atom_A(w)\right\}\nonumber
\end{align}
In this system, the top equation ensures that the measures of all states in~$\Xfrak_{w,A}$ add up to~$1$, and the bottom equations guarantee that $\Prob_A(\alpha)$'s are obtained as the sums of measures of states where $\alpha$'s are true. Moreover, we note that it takes linear time w.r.t.\ $\lmc(\phi)$ to determine the value of a~given~$a$ and that \eqref{equ:biglineq}~contains $|\atom_A(w)|+1$ equations. Thus, applying~\cite[Theorem~9.3]{Chvatal1983}, we can guess a~list $\Lmsf$ containing $|\atom_A(w)|+1$ variables $u^A_\Xi$ s.t.\ $\Xi$'s~are maximal consistent sets and for every $\alpha\in\atom_A(w)$, there is some~$\Xi$ s.t.~$\alpha\in \Xi$. Then we set $u^A_\Xi=0$ for every $u^A_\Xi\notin\Lmsf$. This gives us the following system of equations whose size is \emph{polynomial} w.r.t.\ $\atom_A(w)$:
\begin{align}\label{equ:smalllineq}
\left\{\sum_{u^A_\Xi\,\in\,\Lmsf}u^A_\Xi=1\right\}&\cup\tag{$\mathsf{small}^A_w$}\\
\left\{\sum_{u^A_\Xi\,\in\,\Lmsf}(a_{\Xi,\alpha}\cdot u^A_\Xi)=y^A_\alpha~\middle|~\alpha\in\atom_A(w)\right\}\nonumber
\end{align}
We repeat this process, obtaining \eqref{equ:smalllineq}'s for every $A\in\Ag$. Then, we consider the following system of polynomial inequalities:
\begin{align*}
\Smsf_w&=\Bmc^\Imbb_w\cup\bigcup_{A\in\Ag}\eqref{equ:smalllineq}\cup\bigcup_{\myoverset{\alpha\in\atom_A(w)}{A\in\Ag}}\left\{y^A_\alpha=x_{w{:}\Prob_A(\alpha)}\right\}
\end{align*}
Note that the size of~$\Smsf_w$ is polynomial in~$\lmc(\phi)$. Thus, it takes polynomial space to determine whether $\Smsf_w$~has solutions. If it \emph{does not}, then we close~$\Bmc$, and pick the next branch. If \emph{it does}, we mark the constraint that generated~$w$ (say, $w':\ActBox{a}\psi\leq P$) as ‘safe’, delete $w'\Rmsf_aw$, $\Smsf_w$, and all formulas labelled with~$w$ from~$\Bmc$. We then pick another formulaic constraint of the form $w':\blacksquare\psi\leq P'$ and repeat the process. Once all constraints of the form $w':\blacksquare\psi\leq P'$ are ‘safe’, we mark~$w'$ as ‘safe’. We proceed like this until $w_0$~is marked ‘safe’ (in which case, $\phi$~is $\preafb$-satisfiable) or all branches of the tableau are closed (i.e., $\phi$~is unsatisfiable).
\end{proof}

We finish the section with a~brief remark.
\begin{remark}\label{rem:WMP}
Recall from~\cite[Lemma~3]{Hajek2005} and~\cite[Lem\-ma~4.6]{Vidal2022} that modal Łukasiewicz logic has the~\emph{witnessed model property}. A~witnessed model is a~model $\Mfrak$ where $\Imc_\Mfrak(\ActBox{a}\phi,w)=x$ (or $\Imc_\Mfrak(\knows_A\phi',w')=x'$, respectively) implies the existence of $s\in \Rmsf_a(w)$ ($s'\in \Emsf_A(w)$) s.t.\ $\Imc_\Mfrak(\phi,s)=x$ ($\Imc_\Mfrak(\phi',s')=x'$). In particular, this, together with Remark~\ref{rem:fewermodels}, means that the tableaux rules (Fig.~\ref{fig:tableauxrules}) are sound w.r.t.\ $\{\bullet,\rightarrow_\Pi\}$-free $\Lpr$~formulas. Moreover, it follows that a~$\{\bullet,\rightarrow_\Pi\}$-free $\phi\in\Lpr$ is $\preafb$-satisfiable iff it is $\prea$-satisfiable. Thus, witnessed model property \emph{implies finite model property}.
\end{remark}

Thus, the next statement follows.
\begin{restatable}{theorem}{additivepreapspace}\label{theorem:additivepreapspace}
Given a~$\{\bullet,\rightarrow_\Pi\}$-free $\phi\in\Lpr$, it is $\pspace$-complete to determine whether it is $\prea$-satisfiable.
\end{restatable}
\section{Polynomial Fragments of $\prea$\label{sec:polynomial}}
As is well-known, classical logic has many polynomially decidable fragments. One of the most important such fragments is \emph{Horn logic}. A~\emph{Horn theory} can be represented as a~set of \emph{rules} of two kinds: $\bigwedge^n_{i=1}p_i{\rightarrow}q$ and $\bigwedge^n_{i=1}p_i{\rightarrow}\bot$. In Łukasiewicz logic, there is a~similar polynomially decidable fragment studied by Bofill et al.~\shortcite{BofillManyaVidalVillaret2019}. It is called \emph{simple-clause fragment}: a~simple clause theory is a~theory that contains clauses of the form $\bigoplus^m_{i=1}p_i\oplus\bigoplus^n_{j=1}\neg q_j$. Using Definition~\ref{def:preasemantics} and Convention~\ref{conv:boxesdiamonds}, simple clauses can also be represented in the rule form.

In this section, we will consider two modal probabilistic expansions of the simple clause fragment. The difference is that instead of propositional literals, our clauses (rules) will be built over \emph{uniform probabilistic literals}.
\begin{definition}[Uniform modal literals ($\uml$'s)]\label{def:UML}
Consider the following grammar that defines a~fragment of~$\Lev$:
\begin{align*}
\pi&\coloneqq p\mid\neg p\mid\ActLozenge{a}\pi\mid\knowsLozenge{A}\pi\\
\nu&\coloneqq p\mid\neg p\mid\ActBox{a}\nu\mid\knows_A\nu
\end{align*}
Formulas of the type $\pi$ and $\nu$ are called \emph{$\blacklozenge$-literals} and \emph{$\blacksquare$-literals}, respectively. A~literal is \emph{proper} if it contains at least one modality and \emph{propositional} otherwise.
\end{definition}
\begin{definition}[Uniform probabilistic literals ($\upl$'s)]\label{def:UPL}
Consider the following grammar ($\lambda$ is a~$\neg$-free $\uml$):
\begin{align*}
\zeta&\coloneqq\Prob_A(\lambda)\mid\neg\Prob_A(\lambda)\mid\ActLozenge{a}\zeta\mid\knowsLozenge{A}\zeta\\
\eta&\coloneqq\Prob_A(\lambda)\mid\neg\Prob_A(\lambda)\mid\ActBox{a}\eta\mid\knows_A\eta
\end{align*}
Formulas of the type $\zeta$ and $\eta$ are called \emph{$\blacklozenge^\Prob$-literals} and \emph{$\blacksquare^\Prob$-literals}, respectively. A~literal is \emph{proper} if it contains at least one modality \emph{on the outer layer} and \emph{propositional} otherwise.
\end{definition}

Observe that all modal events from Example~\ref{example:modalevents} are expressed via $\uml$'s and can occur in $\upl$'s. Furthermore, it is important to note that $\blacksquare^\Prob$-literals can contain $\blacklozenge$, and $\blacklozenge^\Prob$-literals can contain $\blacksquare$ on the inner layer. Moreover, the fragments of~$\Lpr$ that we are considering in this section will not use connectives~$\bullet$ and~$\rightarrow_\Pi$. Thus, again (cf.~Remark~\ref{rem:WMP}), by~\cite[Lemma~3]{Hajek2005} and~\cite[Lem\-ma~4.6]{Vidal2022}, it follows that their $\prea$-satisfiability can be w.l.o.g.\ established w.r.t.\ \emph{witnessed models} and using rules from Fig.~\ref{fig:tableauxrules}.
\subsection{Universal Probabilistic Rules\label{ssec:UPR}}
Let us now define our first fragment~--- \emph{universal probabilistic rules}. As the name suggests, we can only use box-like modalities on the outer layer. Note, however, that \emph{inner formulas} can use both types of modalities.
\begin{definition}[Universal probabilistic rules]\label{def:UPR}
A \emph{universal probabilistic rule} ($\uprobrule$) is
\begin{itemize}[noitemsep,topsep=0pt]
\item either a~formula of the form $\bigodot^n_{i=1}\eta_i\rightarrow\eta'$ s.t.\ $\eta$'s are \emph{proper} $\blacksquare^\Prob$-literals;
\item or a~formula of the form $\bigodot^n_{i=1}\eta_i\rightarrow\zero$ s.t.\ $\eta$'s are $\blacksquare^\Prob$-literals.
\end{itemize}
A~\emph{$\uprobrule$ theory} is a~set of formulas containing $\uprobrule$'s.
\end{definition}
For simplicity, we will write rules $\eta\rightarrow\zero$ as $\blacklozenge^\Prob$-literals. E.g., $\ActLozenge{a}\neg\Prob_A(p)$ instead of $\ActBox{a}\Prob_A(p)\rightarrow\zero$.

Let us briefly illustrate the expressiveness of universal probabilistic rules. We recall several statements considered in Section~\ref{sec:expressivity}. One can see we can \emph{compare} two probabilities: $\Prob_A (\ActBox{a}p)\to\Prob_A(\ActBox{b}p)$ is a~$\uprobrule$. We can also express that performing~$a$ does not \emph{increase} the probability with $\ActBox{a}\Prob_\mathtt{ob} (\knows_A p) \to \Prob_\mathtt{ob}(\knows_A p)$. Moreover, $\onehalf \to \Prob_\mathtt{ob}(\knows_A p)$ can be rewritten as follows: $\{\Prob_\mathtt{ob}(q)\rightarrow\neg\Prob_\mathtt{ob}(q),\neg\Prob_\mathtt{ob}(q)\rightarrow\Prob_\mathtt{ob}(q),\Prob_\mathtt{ob}(q)\rightarrow\Prob_\mathtt{ob}(\knows_A p)\}$.

On the other hand, we cannot express that an action does not \emph{decrease} a~given probability: $\Prob_\mathtt{ob}(\ActBox{a}p) \to\ActBox{b} \Prob_\mathtt{ob}(\ActBox{a}p)$ is not a~$\uprobrule$ since $\Prob_\mathsf{ob}(\ActBox{a}p)$ is not a~proper $\blacksquare^\Prob$-literal. Likewise, qualitative uncertainty~--- $\knows_A \Prob_\mathtt{ob}(p)\leftrightarrow \knowsLozenge{A}\Prob_\mathtt{ob}(p)$~--- is not expressible as a~$\uprobrule$.

Furthermore, crisp statements with~$\triangle$ rational constants other than~$\onehalf$ (cf.\ Examples~\ref{example:bidding1} and~\ref{example:bidding2}) are not expressible either. On the other hand, all formulas in Example~\ref{example:bidding3} belong to the $\uprobrule$ fragment.

The next statement establishes the polynomial decidability of $\uprobrule$ theories. Observe from Remark~\ref{rem:WMP} that since $\uprobrule$ theories \emph{do not contain} $\bullet$ and $\rightarrow_\Pi$, our decidability result holds for \emph{all} probabilistic models.
\begin{restatable}{theorem}{UPRpolynomial}\label{theorem:UPRpolynomial}
Given a~$\uprobrule$ theory~$\Gamma$, it takes polynomial time to verify whether it is $\prea$-satisfiable.%
\end{restatable}
Let us sketch the proof. We first show that it takes polynomial time to decide whether a~set of $\upl$'s is satisfiable. Then, we consider \emph{irreducible theories}, i.e., $\uprobrule$ theories $\Gamma=\Theta\cup\Xi$ where $\Theta$ is a~set of~$\upl$'s and $\Xi$ is a~set of $\uprobrule$'s with at least two literals each s.t.\ $\Xi$~does not contain any literal~$\theta$ s.t.~$\Theta\models_{\prea}\neg\theta$. We show that it takes polynomial time to verify whether a~given $\uprobrule$ theory is irreducible.

We then prove that every irreducible theory is $\prea$-satisfiable. First, note that $\Theta$~is satisfiable. Second, we pick one \emph{proper} $\blacksquare^\Prob$-literal from the left-hand side of each rule in~$\Xi$ (that contains proper $\blacksquare^\Prob$-literals) and show that there is a~pointed probabilistic model $\langle\Mfrak,w\rangle$ s.t.\ $\Imc_\Mfrak(\theta,w)=1$ for all $\theta\in\Theta$, and every literal we picked has value~$0$. Third, we show that the propositional literals in the rules that are not satisfied by that model can be evaluated at~$\tfrac{1}{2}$ in~$w$ (whence all rules \emph{without} proper $\blacksquare^\Prob$-literals have value~$1$ at~$w$). This guarantees us that every formula in~$\Gamma$ has value~$1$ at~$w$.

Let us now sketch a~polynomial decision procedure for $\uprobrule$ theories. First, we transform the theory into the clause form, using that $\Imc_\Mfrak(\phi\rightarrow\chi,w)=\Imc_\Mfrak(\neg\phi\oplus\chi,w)$ in each probabilistic model~$\Mfrak$. Second, we apply the following rule:
\begin{align*}
\text{if }\Theta\models_{\prea}\neg\theta,&\text{ then }\Theta,\theta\oplus\kappa\models_{\prea}\kappa
\end{align*}
Using this rule, we need polynomial time to either transform~$\Gamma$ into an irreducible theory or infer a~literal $\theta'$ s.t.\ \mbox{$\Theta',\theta'\models_{\prea}\zero$} with $\Theta'$ being (possibly) augmented with other $\upl$'s inferred via the above transformation. In the first case, $\Gamma$~is $\prea$-satisfiable, and in the second, it is not.
\subsection{Existential Probabilistic Rules\label{ssec:EPR}}
Our second fragment is \emph{existential} probabilistic rules. One can think of them as a~probabilistic version of the Łukasiewicz description logic $\mathcal{EL}_\bot$ (cf., e.g., \cite[Ch.~6]{BaaderHorrocksLutzSattler2017} for a~presentation of \emph{classical}~$\mathcal{EL}_\bot$). The main difference, of course, is that we will consider the complexity of \emph{local} satisfiability (whether a~formula is true in some state of a~model) of theories containing existential probabilistic rules. This is because \emph{global} satifiability (whether a~formula is true in all states of a~model) of Łukasiewicz $\mathcal{EL}_\bot$ is undecidable~\cite{BaaderPenaloza2011,CeramiStraccia2013,BorgwardtCeramiPenaloza2015,BorgwardtCeramiPenaloza2017,BaaderBorgwardtPenaloza2017,BobilloStraccia2018}.
\begin{definition}\label{def:EPR}
An \emph{existential positive probabilistic rule} ($\eprplus$) is a~formula of one of the following kinds:
\begin{itemize}[noitemsep,topsep=0pt]
\item a~$\upl$;
\item $\bigodot^n_{i=1}\zeta_i{\rightarrow}\zero$ with $\zeta$'s being propositional literals or $\neg$-free $\blacklozenge^\Prob$-literals;
\item $\bigodot^n_{i=1}\zeta_i\rightarrow\zeta$ with $\zeta$'s being \emph{proper} $\neg$-free $\blacklozenge^\Prob$-literals.
\end{itemize}
\end{definition}
There is also a~syntactic difference between $\eprplus$ and $\mathcal{EL}_\bot$. Namely, we can freely use negation over propositional atoms. E.g., a~rule $\neg\Prob_A(\ActBox{a}p){\odot}\neg\Prob_B(\knows_A\knows_B q){\rightarrow}\zero$ is allowed, despite the negation in the left-hand side.

On the other hand, $\eprplus$'s and $\uprobrule$'s are not completely dual. In particular, in contrast to $\uprobrule$ theories, we cannot use negation in $\blacklozenge^\Prob$-literals when they occur in implications. Thus, existential rules such as $\ActLozenge{a}\neg\Prob_A(p)\rightarrow\knowsLozenge{A}\Prob_B(q)$ are not allowed (while $\ActBox{a}\neg\Prob_A(p)\rightarrow\knows_A\Prob_B(q)$ is a~$\uprobrule$). Hence, no formula from Examples~\ref{example:bidding1}--\ref{example:bidding3}, except $\Prob_A(\ActBox{m}B\,\texttt{folds})\rightarrow\Prob_A(\ActBox{a}B\,\texttt{folds})$, belongs to the $\eprplus$ fragment.

The next statement shows that $\prea$-satisfiability of $\eprplus$ theories is also decidable in polynomial time. The proof is similar to that of Theorem~\ref{theorem:UPRpolynomial}.
\begin{restatable}{theorem}{EPRpolynomial}\label{theorem:EPRpolynomial}
Given an $\eprplus$ theory~$\Gamma$, it takes polynomial time to verify whether it is $\prea$-satisfiable.%
\end{restatable}
We conclude the section with a~short sketch of the proof. Observe first from Theorem~\ref{theorem:UPRpolynomial} that it takes polynomial time to check the $\prea$-satisfiability of a~given set of $\upl$'s or whether an $\upl$ is entailed from a~set of $\upl$'s.

Now, as was the case in Theorem~\ref{theorem:UPRpolynomial}, it remains to show that irreducible $\eprplus$ theories are $\prea$-satisfiable. The difference is that instead of picking $\blacksquare^\Prob$-literals from each rule and assigning them value~$0$ at~$w$, we show that there is a~pointed probabilistic model $\langle\Mfrak,w\rangle$ where one proper \emph{$\blacklozenge^\Prob$-literal} from the left-hand side of each rule has value~$0$ at~$w$. Thus, given an $\eprplus$ theory, we can use the same procedure as in Theorem~\ref{theorem:UPRpolynomial} to either transform it into an irreducible (and hence, satisfiable) theory or infer a~contradiction.
\section{Related Work\label{sec:relatedwork}}
As we mentioned in the introduction, the first family of logics that combined modal (epistemic) and probabilistic reasoning was proposed by Fagin and Halpern~\shortcite{FaginHalpern1994}. There are several important differences between their logic and ours. First, probabilistic formulas in~\cite{FaginHalpern1994} are expressed via linear inequalities that are \emph{two-valued}, and thus, are not well suited to express imprecise statements such as ‘the probability of~$\alpha$ \emph{is not much} lower (higher) than the probability of~$\beta$’. Second, the expressivity of the modal probabilistic language proposed by Fagin and Halpern is incomparable with that of~$\Lpr$. On the one hand, there are no separate \emph{inner} and \emph{outer} formulas. Thus, e.g., the following formula (adapted to the notation in this paper) is permitted: $(\Prob_B(\knows_A(\Prob_A(a){\geq}\overline{\tfrac{2}{3}})){\leq}\onehalf){\rightarrow}\knows_Bq$. On the other hand, inequality formulas can contain only one agent label. Thus, statements such as ‘$A$~and~$B$ consider $\alpha$ equally likely’ ($\triangle(\Prob_A(\alpha){\leftrightarrow}\Prob_B(\alpha))$ in~$\Lpr$) cannot be expressed, unless one specifies the degree of likelihood. Furthermore, the validity of some probabilistic logics is $\mathsf{EXPTIME}$-complete~\cite[Theorem~4.5]{FaginHalpern1994}. Observe, however, that all logics in~\cite{FaginHalpern1994} have the finite model property, while in our case, this holds only for the $\{\bullet,\rightarrow_\Pi\}$-free fragment of~$\Lpr$ (recall Remark~\ref{rem:WMP}).

Our work, however, continues the study of \emph{fuzzy} modal probabilistic logics. In particular, we expand the fuzzy \emph{epistemic} probabilistic logic proposed by Corsi et al.~\shortcite{CorsiFlaminioGodoHosni2023} with \emph{action} modalities as well as with additional connectives~$\bullet$ and~$\rightarrow_\Pi$. These allow us to express (subjective) \emph{conditional probability} of~$\alpha$ given~$\beta$ via $\Prob_A(\beta)\to_\Pi\Prob_A(\alpha\land\beta)$ as well as \emph{degrees of independence} of~$\alpha$ and~$\beta$: $\Prob_A(\alpha\wedge\beta)\leftrightarrow(\Prob(\alpha)\bullet\Prob(\beta))$. Moreover, Corsi et al.\ study the axiomatisation and expressivity of their epistemic probabilistic logic, while we are concentrating on complexity. We also note that the logic in~\cite{CorsiFlaminioGodoHosni2023} is the fragment of~$\prea$ without~$\bullet$, $\rightarrow_\Pi$, and action modalities.

Another closely related work is by Majer and Sedl\'{a}r~\shortcite{MajerSedlar2025} and Kozhemiachenko and Sedl\'{a}r~\shortcite{KozhemiachenkoSedlar2025}. There, a~logic $\KLukProdProb$ with action modalities and propositional event formulas is considered. Note, however, that even though $\prea$ is more expressive than $\KLukProdProb$, the satisfiability problem in both logics is $\pspace$-complete. Thus, in this paper, we combined the frameworks of~\cite{CorsiFlaminioGodoHosni2023} and~\cite{KozhemiachenkoSedlar2025}, by proposing a~logic with both epistemic and action modalities on its outer and inner layers.

Finally, we note that to the best of our knowledge, there are no results on the complexity (and decidability) of Łukasiewicz modal logic extended with~$\triangle$ over arbitrary frames (this logic fails the FMP as we observed in the end of Section~\ref{sec:expressivity}), let alone of the modal logic over $\LukProd$ (combined Łukasiewicz and Product modal logic). Thus, there are no results on the decidability of their probabilistic expansions. Still, there is an interesting pattern: usually, the complexity of a two-layered logic built over logics~$\Lmbf_1$ and~$\Lmbf_2$ coincides with that of the harder logic. E.g., $\mathsf{FP}(\LukProd)$ by Hajek and Tulipani~\shortcite{HajekTulipani2001} uses events expressed as Boolean formulas and is in~$\pspace$, just as the propositional~$\LukProd$. Likewise, the satisfiability of arbitrary probabilistic atoms~$\Prob_A(\alpha)$ is $\pspace$-hard because (classical) modal logic is $\pspace$-hard. Thus, we conjecture that the complexity of $\prea$ (over arbitrary frames) will coincide with that of $\EALukProd$.
\section{Conclusion and Future Work\label{sec:conclusion}}
We introduced and studied $\prea$, a~modal probabilistic logic that unifies and expands the frameworks of Corsi et al.~\shortcite{CorsiFlaminioGodoHosni2023} on one hand and Majer and Sedl\'ar~\shortcite{MajerSedlar2025} and Kozhemiachenko and Sedl\'ar~\shortcite{KozhemiachenkoSedlar2025}, on the other hand. A~distinguishing feature of our framework is the presence of epistemic and action modalities in both logical layers.

We showed that $\prea$ can express probabilistic statements about modal events involving knowledge or the effects of actions and lower and upper probabilities of events arising from epistemic uncertainty or from uncertainty about outcomes of actions. Such specifications are  important in domains such as robotics and cybersecurity. We established that the satisfiability problem for the full language over finitely-branching frames is $\pspace$-complete. We also identified specific syntactic fragments for which validity is decidable in polynomial time.

There are several promising avenues for future research. A primary objective is to provide a sound and complete axiomatisation for $\prea$. To enrich the framework and improve its applicability to reasoning about programs and multi-agent systems, we also intend to extend the current language to incorporate Propositional Dynamic Logic constructs and operators for group knowledge.

Finally, we plan to explore the semantic boundaries and other computational aspects of the logic. Our current complexity results in Section~\ref{sec:decidability} for the full $\Lpr$ rely on a~restriction to finitely branching frames. Investigating the complexity of the logic without this constraint is a~necessary step towards generalising the framework. We also intend to expand the scope of uncertainty representation by incorporating fuzzy events and exploring other uncertainty measures alongside their qualitative counterparts. Identifying additional polynomially-decidable fragments of $\prea$ remains a~priority for practical implementation.

\vfill\eject
\section*{Acknowledgements}
Igor Sedlár was supported by the grant 22-16111S of the Czech Science Foundation. The authors would also like to thank the reviewers for their constructive suggestions that improved the quality of the paper.
\section*{AI Declaration}
The authors used Google Gemini to obtain feedback on the formulations of Examples in Section 3.
\bibliographystyle{kr}
\bibliography{kr-sample}
% \end{document}
\clearpage
\onecolumn
\appendix
\section{Proofs of Section~\ref{sec:decidability}}
\tableauxcompleteness*
\begin{proof}
The proof is standard, so we provide only a~sketch.

For the soundness part, we observe first that closed branches do not have realising models. Indeed, if $\Bmc$ is closed, then $\Bmc^\Imbb$ has no solutions, whence $\Bmc$ cannot have a~realising model. Now, it suffices to show that if $\Mfrak$ realises the premises of a~rule, it also realises the conclusion. For propositional rules, this follows straightforwardly from Definition~\ref{def:preasemantics}. Let us now consider the $\Box_\leq$ rule as an example. Assume that $\Mfrak$ realises $w:\ActBox{a}\chi\leq P$. Then $\Imc_\Mfrak(\ActBox{a}\chi,w)\leq P$ and there must be some $w'\in \Rmsf_a(w)$ s.t.\ $\Imc_\Mfrak(\chi,w')\leq P$. Hence, the conclusion is also realised. Rules $\Box_\geq$, $\knows_\leq$, and $\knows_\geq$ can be dealt with similarly.

Let us now consider the completeness. It suffices to show that every complete open branch $\Bmc$ has a~realising model. Let $\Smc_{\Bmc^\Imbb}$ be a~solution of the system of inequalities $\Bmc^\Imbb$ that corresponds to~$\Bmc$. We construct the model $\Mfrak=\langle W,\Emsf,\Rmsf,v\rangle$ as follows:
\begin{align*}
W&=\{w\mid w\text{ occurs on }\Bmc\}\\
\Rmsf_a&=\{\langle w,w'\rangle\mid w\Rmsf_aw'\in\Bmc\}\tag{for every $a\in\Ac$}\\
\Emsf_A&=\{\langle w,w'\rangle\mid\exists\cluster^A:w\in\cluster^A\text{ and }w'\in\cluster^A\text{ occur on~}\Bmc\}\tag{for every $A\in\Ag$}\\
v(p,w)&=\Smc_{\Bmc^\Imbb}(x_{w:p})
\end{align*}
It remains to show that $\Imc_\Mfrak(\chi,w)=\Smc_{\Bmc^\Imbb}(x_{w:\chi})$ for every $\chi$ occurring on~$\Bmc$.

We proceed by induction on~$\chi$. The base case of $\chi\in\Var$ holds by construction of~$\Mfrak$. The propositional cases can be verified by a~straightforward application of the induction hypothesis. Let us now consider the case of $\chi=\ActBox{a}\psi$. Assume that $w{:}\ActBox{a}\psi\leq P\in\Bmc$. As $\Bmc$ is complete, $\{wRw',w'{:}\psi\leq P\}\subseteq\Bmc$. By the induction hypothesis, $\Mfrak$~realises $w'{:}\psi\leq P$. Thus, there is $w'\in \Rmsf_a(w)$ s.t.\ $\Imc_\Mfrak(\psi,w')\leq P$. It follows that $\Imc_\Mfrak(\ActBox{a}\psi,w)\leq P$, as required. Now assume that $w{:}\ActBox{a}\psi\geq P\in\Bmc$. If there is no $u$ s.t.\ $w\Rmsf_au\in\Bmc$, it means that $\Rmsf_a(w)=\varnothing$, whence $\Imc_\Mfrak(\ActBox{a}\psi,w)=1$. Thus, $w{:}\ActBox{a}\psi\geq P$ is realised. If there is some $u$ s.t.\ $w\Rmsf_au\in\Bmc$, then for each such $u$, $u{:}\psi\geq P\in\Bmc$. Again, by the induction hypothesis, all these constraints are realised, whence $\Imc_\Mfrak(\ActBox{a}\psi,w)\leq P$, as required. The case of $\chi=\knows_A\psi$ can be dealt with in a~similar manner.
\end{proof}
\SIequivalence*
\begin{proof}
It is clear that if $\phi$ is valid on SI-models, then it is valid on standard models. For the converse direction, let $\phi$ be \emph{not valid} on SI-models. Then $\Imc_\Mfrak(\phi,w)<1$ for some SI-model $\Mfrak=\langle W,\Emsf,\Rmsf,\Xfrak,\mu,V\rangle$ with $\Xfrak=\langle W_\Xfrak,\Emsf_\Xfrak,\Rmsf_\Xfrak,V\rangle$. We define $\Mfrak'=\langle W',\Fmbf',\Emsf',\Rmsf',\mu',V'\rangle$:
\begin{itemize}[noitemsep,topsep=0pt]
\item $W'=W\uplus W_\Xfrak$;
\item $\Fmc_{w,A}'=2^{W'}$ for all $w\in W'$ and $A\in\Ag$;
\item $\Emsf'_A=\Emsf_A\uplus \Emsf_{A_\Xfrak}$ for every $A\in\Ag$;
\item $\Rmsf'_a=\Rmsf_a\uplus\Rmsf_{a_\Xfrak}$ for every $a\in\Ac$;
\item $\mu'_{w,A}(\{s\})=
\begin{cases}
\mu_{w,A}(\{s\})&\text{if }s\in W_\Xfrak\\
0&\text{otherwise}
\end{cases}$;
\item $V'=V$
\end{itemize}
Observe now that for every $\alpha\in\Lev$, $A\in\Ag$, and $w\in W$, $\mu_{w,A}(\|\alpha\|_\Mfrak)=\mu'_{w,A}(\|\alpha\|_{\Mfrak'})$. Let us now show by induction of $\phi\in\Lpr$ that $\Imc_\Mfrak(\phi,w)=\Imc_{\Mfrak'}(\phi,w)$ for every $w\in W$.

Let $\phi=\Prob_A(\alpha)$ with $\alpha\in\Lev$. Observe that epistemic and action relations in~$\Mfrak'$ are disjoint unions of those in $\langle W,\Emsf,\Rmsf\rangle$ and $\langle W_\Xfrak,\Emsf_\Xfrak,\Rmsf_\Xfrak\rangle$ and $V'=V$. Thus, we have that $s\in\|\alpha\|_{\Mfrak'}\cap W_\Xfrak$ iff $s\in\|\alpha\|_\Xfrak$ for every $\alpha\in\Lev$. Furthermore, the only states in~$\Mfrak'$ that can have a~positive measure are exactly the states that were in~$W_\Xfrak$. It now follows that $\mu'_{w,A}(\|\alpha\|_{\Mfrak'})=\mu_{w,A}(\|\alpha\|_\Xfrak)$. The basis case now follows.

The cases of propositional connectives $\{\neg,\rightarrow,\bullet,\rightarrow_\Pi\}$ follow by a~straightforward application of the induction hypothesis. Let us now consider the case of $\phi=\knows_A\chi$. Observe from the construction of $\Mfrak'$ that $\{t\mid s\Emsf'_At\}=\{t\mid s\Emsf_At\}$ for every $s\in W$ and $A\in\Ag$. By the induction hypothesis, we have that $\Imc_{\Mfrak'}(\chi,w')=\Imc_\Mfrak(\chi,w')$ for every $w'$ s.t.~$w\Emsf'_Aw'$. It follows that $\Imc_{\Mfrak'}(\knows_A\chi,w)=\Imc_\Mfrak(\knows_A\chi,w)$. The case of $\phi=\ActBox{a}\chi$ can be dealt with in a~similar manner.
\end{proof}

Before proceeding to the proof of Theorem~\ref{theorem:generalcasedecidability}, we introduce several technical notions.
\begin{definition}\label{def:eventsat}
Given $\alpha\in\Lev$, we say that:
\begin{itemize}[noitemsep,topsep=0pt]
\item $\alpha$ is \emph{$\EA$-satisfiable} iff there is an event model $\Efrak$ and $w\in\Efrak$ s.t.\ $\Efrak,w\vDash\alpha$;
\item $\alpha$ is \emph{$\EA$-valid} iff for every event model $\Efrak$ and every $w\in\Efrak$, it holds that $\Efrak,w\vDash\alpha$.
\end{itemize}
Given $\Gamma\subseteq\Lev$, we will say that it is $\EA$-satisfiable iff there is an event model $\Efrak$ and $w\in\Efrak$ s.t.\ $\Efrak,w\vDash\gamma$ for all $\gamma\in\Gamma$.
\end{definition}
\begin{definition}[Maximal consistent sets]\label{def:MCS}
Let $\Gamma\subseteq\Lev$ be closed under negated subformulas, i.e., satisfy the following property:
\begin{align*}
\forall\alpha:(\exists\beta:\alpha\text{ is a subformula of }\beta~\&~\beta\in\Gamma)\Rightarrow\alpha\in\Gamma\text{ or }\neg\alpha\in\Gamma
\end{align*}
A~\emph{maximal consistent subset of}~$\Gamma$ ($\Gamma$-MCS) is a~set $\Xi$ that satisfies the following property:
\begin{align*}
\forall\alpha\in\Gamma:\alpha\in\Xi\Leftrightarrow\Xi\cup\{\alpha\}\text{ is $\EA$-satisfiable}
\end{align*}
\end{definition}
\begin{definition}\label{def:MCSmodels}
Let $\Gamma\subseteq\Lev$ be closed under negated subformulas. An \emph{MCS event model for~$\Gamma$} is the following structure $\Efrak^\Gamma=\langle W^\Gamma,\Emsf^\Gamma,\Rmsf^\Gamma,V^\Gamma\rangle$:
\begin{align*}
W^\Gamma&=\{\Xi\subseteq\Gamma\mid\Xi\text{ is a }\Gamma\text{-MCS}\}\\
\Emsf^\Gamma_A&=\{\langle\Delta,\Theta\rangle\mid\forall\knows_A\alpha\in\Gamma:\knows_A\alpha\in\Delta\Leftrightarrow\knows_A\alpha\in\Theta\}\tag{$A\in\Ag$}\\
\Rmsf^\Gamma_a&=\{\langle\Delta,\Theta\rangle\mid\forall\ActBox{a}\alpha\in\Gamma:\ActBox{a}\alpha\in\Delta\Rightarrow\alpha\in\Theta\}\tag{$a\in\Ac$}\\
V^\Gamma(p)&=\{\Xi\in W^\Gamma\mid p\in\Xi\}
\end{align*}
\end{definition}

\begin{lemma}\label{lemmaMCSmodelswelldefined}
For every $\Gamma\neq\varnothing$~that is closed under negated subformulas, $\Efrak^\Gamma$ is an event model.
\end{lemma}
\begin{proof}
We note briefly that $W^\Gamma\neq\varnothing$ and that $\Emsf^\Gamma_A$'s are equivalence relations in~$\Gamma$ as one can see from Definition~\ref{def:MCSmodels}. The statement now follows.
\end{proof}

The next statement can now be shown by induction on~$\alpha\in\Lev$ same as~\cite[Theorem~4.1]{FaginHalpern1994}.
\begin{lemma}\label{lemma:MCSmodels}
Let $\Gamma\subseteq\Lev$ be closed under negated subformulas and $\Efrak^\Gamma=\langle W^\Gamma,\Emsf^\Gamma,\Rmsf^\Gamma,V^\Gamma\rangle$ be an MCS event model for~$\Gamma$. Then the following holds:
\begin{align*}
\forall\phi\in\Gamma,\Xi\text{ is a }\Gamma\text{-MCS}:\phi\in\Xi\Leftrightarrow\Efrak^\Gamma,\Xi\vDash\phi
\end{align*}
\end{lemma}
\begin{proof}
We proceed by induction on~$\alpha\in\Lev$. The basis case of $\alpha=p$ for some $p\in\Var$ follows by Definition~\ref{def:MCSmodels}. The cases of $\alpha=\beta\wedge\gamma$ and $\alpha=\sneg\beta$ can be obtained by a~simple application of the induction hypotheses.

Now, consider $\alpha=\knows_A\beta$. If $\knows_A\beta\in\Xi$, then for every $\Theta$ s.t.\ $\Emsf^\Gamma_A(\Xi,\Theta)$ holds, we have that $\beta\in\Theta$, whence by the induction hypothesis, we have that $\Efrak^\Gamma,\Theta\vDash\beta$ for all $\Theta$ s.t.\ $\Emsf^\Gamma_A(\Xi,\Theta)$. Hence, $\Efrak^\Gamma,\Xi\vDash\knows_A\beta$.

Conversely, assume that $\Efrak^\Gamma,\Xi\vDash\knows_A\beta$. Set $\Xi^A=\{\knows_A\gamma\mid\knows_A\gamma\in\Xi\}\cup\{\sneg\knows_A\gamma\mid\neg\knows_A\gamma\in\Xi\}$. We show that $\knows_A\beta\in\Xi^A$. For that observe that $\Xi^A\cup\{\sneg\beta\}$ is \emph{$\EA$-unsatisfiable}. Indeed, otherwise there would be some $\Theta\in W^\Gamma$ s.t.\ $\Xi^A\cup\{\sneg\beta\}\subseteq\Theta$. But in such case, by Definition~\ref{def:MCSmodels}, we would have that $\Emsf_A(\Xi,\Theta)$ holds and, furthermore, that $\Efrak,\Theta\nvDash\beta$, by the induction hypothesis. Now, since $\Xi^A\cup\{\sneg\beta\}$ is $\EA$-unsatisfiable, we have that $\knows_A(\bigwedge_{\xi\in\Xi^A}\xi\rightarrow\beta)$ is $\EA$-valid. Thus, $\knows_A(\bigwedge_{\xi\in\Xi^A}\xi)\rightarrow\knows_A\beta$ is $\EA$-valid. Now, observe that $\bigwedge_{\xi\in\Xi^A}\xi\rightarrow\knows_A(\bigwedge_{\xi\in\Xi^A}\xi)$ is also $\EA$-valid. Hence, $\bigwedge_{\xi\in\Xi^A}\xi\rightarrow\knows_A\beta$ is $\EA$-valid, $\Xi\cup\{\knows_A\beta\}$ is satisfiable, and $\knows_A\beta\in\Xi$, as required.

The case of $\alpha=\ActBox{a}\beta$ can be tackled in a~similar manner. If $\ActBox{a}\beta\in\Xi$, then for every $\Theta$ s.t.\ $\Rmsf^\Gamma_a(\Xi,\Theta)$ holds, we have that $\beta\in\Theta$. By the induction hypothesis, we have that $\Efrak^\Gamma,\Theta\vDash\beta$ for all $\Theta$ s.t.\ $\Rmsf^\Gamma_a(\Xi,\Theta)$. Hence, $\Efrak^\Gamma,\Xi\vDash\ActBox{a}\beta$.

Conversely, assume that $\Efrak^\Gamma,\Xi\vDash\ActBox{a}\beta$. Set $\Xi^a=\{\gamma\mid\ActBox{a}\gamma\in\Xi\}$. We show that $\Xi^a\cup\{\sneg\beta\}$ is $\EA$-unsatisfiable. Assume for contradiction that it is \emph{$\EA$-satisfiable}. Then, there would be some $\Theta\in W^\Gamma$ s.t.\ $\Xi^a\cup\{\sneg\beta\}\subseteq\Theta$. But by Definition~\ref{def:MCSmodels}, we would have that $\Rmsf^\Gamma_A(\Xi,\Theta)$ holds. Additionally, by the induction hypothesis, $\Efrak^\Gamma,\Theta\nvDash\beta$. Hence, $\Efrak^\Gamma,\Xi\nvDash\ActBox{a}\beta$, contrary to the assumption. It follows that $\Xi^a\cup\{\sneg\beta\}$ is $\EA$-unsatisfiable. Hence, $\bigwedge_{\xi\in\Xi^a}\xi\rightarrow\beta$ is $\EA$-valid, and so is $\bigwedge_{\xi\in\Xi^a}\ActBox{a}\xi\rightarrow\ActBox{a}\beta$. Thus, $\Xi\cup\{\ActBox{a}\beta\}$ is $\EA$-satisfiable and $\ActBox{a}\beta\in\Xi$.
\end{proof}
\generalcasedecidability*
\begin{proof}
We begin by noting that tableaux for finite sets of formulas always terminate since all rules have the branching factor of at most~$2$ and decompose formulas into their subformulas. We construct a~tableau for $\{w^0_1{:}\phi\geq1\}$ in a~standard way. We begin with propositional rules working with branches in a~\emph{depth-first} manner. Once no more propositional rules are applicable, we begin applying rules for modalities as follows. First, we pick \emph{one} constraint $w_0{:}\ActBox{a}\chi\leq P$ or $w_0{:}\knows_A\chi\leq P$ and apply $\Box_\leq$ or $\knows_\leq$, respectively. This generates a~new state, say, $w^1_{1,1}$ and a~relational term $w_0\Rmsf_aw^1_{1,1}$ or $w^1_{1,1}\in\cluster^A$ (depending on the rule we applied). Using this new state, we apply rules $\Box_\geq$ or $\knows_\geq$ that utilise it as a premise. In that new state, we then apply propositional rules and repeat the process.

It is clear that in at most $\lmc(\phi)$ repetitions we will generate a~state, say $w^m_{1,1}$ with $m\leq\lmc(\phi)$, s.t.\ all constraints $w^m_{1,1}{:}\psi\triangledown P$ are propositional (i.e., do not contain modalities on the \emph{outer layer}). Then we apply propositional rules until all constraints are decomposed into the ones of the form $w^m_{1,1}{:}\Prob_A(\alpha)\triangledown Q$ with $\alpha\in\Lev$. This produces the following system of polynomial inequalities that corresponds to the current branch~$\Bmc$ and state~$w^m_{1,1}$ (below, $\psi$'s are complex formulas s.t.\ $w^m_{1,1}{:}\psi$ occurs on~$\Bmc$, and $P$'s and $Q$'s are polynomials over~$[0,1]$ with integer coefficients):
\begin{align}\label{equ:branch}
\left\{x_{w^m_{1,1}{:}\psi}\triangledown P\mid(w^m_{1,1}{:}\psi\triangledown P){\in}\Bmc\right\}\cup\left\{x_{w^m_{1,1}{:}\Prob_A(\alpha)}\triangledown Q\mid(w^m_{1,1}{:}\Prob_A(\alpha)\triangledown Q){\in}\Bmc\right\}\cup\left\{Q'{\leq}Q''\mid(w^m_{1,1}{:}Q'{\leq}Q''){\in}\Bmc\right\}\tag{$\Bmc^\Imbb_{w^m_{1,1}}$}
\end{align}

For every $A\in\Ag$, we define
\begin{align*}
\atom_A(w^m_{1,1})&=\{\alpha\mid w{:}\Prob_A(\alpha)\triangledown P\in\Bmc\}\\
\subformulas^\sim_A(w^m_{1,1})&=\{\beta\mid\exists\alpha\in\atom_A(w^m_{1,1}):\beta\text{ is a subformula of }\alpha\}\cup\{\neg\beta\mid\exists\alpha\in\atom_A(w^m_{1,1}):\beta\text{ is a subformula of }\alpha\}
\end{align*}
Now, consider an MCS event model for $\subformulas^\sim_A(w^m_{1,1})$ (recall Definition~\ref{def:MCSmodels}) and denote it with $\Xfrak_{w^m_{1,1},A}$. Observe that~$\Xfrak_{w^m_{1,1},A}$ has an at most exponential size w.r.t.~$\lmc(\phi)$.

We now need to ensure that the values of $x_{w^m_{1,1}{:}\Prob_A(\alpha)}$'s are coherent with a~measure on~$\Xfrak_{w^m_{1,1},A}$, i.e., there is a~measure $\mu_{w^m_{1,1},A}$ on~$\Xfrak_{w^m_{1,1},A}$ s.t\ $\mu_{w^m_{1,1},A}(\|\alpha\|)=x_{w^m_{1,1}{:}\Prob_A(\alpha)}$. For that, we introduce new variables $y^A_{\alpha}$ for every $\alpha\in\atom_A(w^m_{1,1})$ and $u^A_\Xi$ for every $\Xi\subseteq\subformulas^\sim_A(w^m_{1,1})$. Here, $y$'s are auxiliary variables that will stand for the values of $\Prob_A(\alpha)$'s at~$w$, and $u$'s are the values of $\mu_{w^m_{1,1},A}(\{\Xi\})$'s (i.e., we encode the \emph{probability distribution} induced by~$\mu_{w^m_{1,1},A}$). Furthermore, we set $a_{\Xi,\alpha}=1$ if $\alpha\in\Xi$ and $a_{\Xi,\alpha}=0$, otherwise. Now, we consider the following system of linear equations:
\begin{align}\label{equ:biglineqappendix}
\underbrace{\left\{\sum_{\Xi\subseteq\subformulas^\sim_A(w)}u^A_\Xi=1\right\}}_{\mathtt{add}}\cup\underbrace{\left\{\sum_{\Xi\subseteq\subformulas^\sim_A(w)}(a_{\Xi,\alpha}\cdot u^A_\Xi)=y^A_\alpha\mid\alpha\in\atom_A(w)\right\}}_{\mathtt{form}}\tag{$\mathsf{big}^A_{w^m_{1,1}}$}
\end{align}
In this system, $\mathtt{add}$ ensures that the measures of all states in~$\Xfrak_{w,A}$, and $\mathtt{form}$ guarantees that $\Prob_A(\alpha)$'s are obtained as the sums of measures of states where $\alpha$'s are true. One can see that \eqref{equ:biglineqappendix}~contains $|\atom_A(w)|+1$ equations

Thus, applying~\cite[Theorem~9.3]{Chvatal1983}, we can guess a~list $\Lmsf$ containing $|\atom_A(w^m_{1,1})|+1$ variables $u^A_\Xi$ s.t.\ $\Xi$'s~are maximal consistent subsets of $\subformulas^\sim_A(w^m_{1,1})$ and for every $\alpha\in\atom_A(w^m_{1,1})$, there is some~$\Xi$ s.t.~$\alpha\in\Xi$. Then we set $u^A_\Xi=0$ for every $u^A_\Xi\notin\Lmsf$. This gives us the following system of equations whose size is \emph{polynomial} w.r.t.\ $\atom_A(w^m_{1,1})$:
\begin{align}\label{equ:smalllineqappendix}
\left\{\sum_{u^A_\Xi\in\Lmsf}u^A_\Xi=1\right\}\cup\left\{\sum_{u^A_\Xi\in\Lmsf}(a_{\Xi,\alpha}\cdot u^A_\Xi)=y^A_\alpha\mid\alpha\in\atom_A(w)\right\}\tag{$\mathsf{small}^A_{w^m_{1,1}}$}
\end{align}
Observe that it takes polynomial space to produce \eqref{equ:smalllineqappendix}. First, the values of a~given~$a$ can be determined in linear time w.r.t.~$\lmc(\phi)$. Second, to check that a~given $\Xi\subseteq\subformulas^\sim_A(w^m_{1,1})$ is a~maximal consistent subset, it suffices to check that $\Xi$ is satisfiable but $\Xi\cup\{\alpha\}$ is \emph{unsatisfiable} for every $\alpha\notin\Xi$. Since formulas in $\subformulas^\sim_A(w^m_{1,1})$ contain only $\sfive$ and $\Kmbf$ modalities, it follows that these checks require only polynomial space.

We repeat this process, obtaining \eqref{equ:smalllineqappendix}'s for every $A\in\Ag$. Then, we consider the following system of polynomial inequalities:
\begin{align*}
\Smsf_{w^m_{1,1}}&=\Bmc^\Imbb_w\cup\bigcup_{A\in\Ag}\eqref{equ:smalllineqappendix}\cup\bigcup_{\myoverset{\alpha\in\atom_A(w)}{A\in\Ag}}\left\{y^A_\alpha=x_{w^m_{1,1}{:}\Prob_A(\alpha)}\right\}
\end{align*}
Note that the size of~$\Smsf_{w^m_{1,1}}$ is polynomial in~$\lmc(\phi)$. Thus, it takes polynomial space to determine whether $\Smsf_{w^m_{1,1}}$~has solutions. If it \emph{does not}, then we close~$\Bmc$, and pick the next branch in the tableau. If \emph{it does}, we delete~\eqref{equ:branch} and~\eqref{equ:smalllineq}, mark the constraint that produced it (i.e., $w^{m-1}_{1,1}{:}\ActBox{a}\chi\leq Q$ or $w^{m-1}_{1,1}{:}\knows_A\chi\leq Q$) as ‘safe’ and also delete its associated relational term ($w^{m-1}_{1,1}\Rmsf_aw^m_{1,1}$ or $w^m_{1,1}\in\cluster^A$, depending on whether we had applied an action rule or an epistemic rule). We then pick the next constraint in $w^{m-1}_{1,1}$ that can produce a~new state (say, $w^{m-1}_{1,1}{:}\ActBox{b}\tau\leq Q'$) and repeat the process. Once all constraints of the form $w^{m-1}_{1,1}{:}\ActBox{a}\chi'\leq Q'$ and $w^{m-1}_{1,1}{:}\knows_A\chi'\leq Q'$ are marked safe, we guess $\mathsf{small}^A_{w^{m-1}_{1,1}}$ as shown above and consider $\Smsf_{w^{m-1}_{1,1}}$. Note, furthermore, that if all probabilistic atoms occurring in $w^{m-1}_{1,1}$ are in the scope of modalities, it suffices to check that $\Bmc^\Imbb_{w^{m-1}_{1,1}}$ has solutions. We repeat the process until either $w^0_1$ is marked safe (in which case, $\phi$ is satisfiable) or all branches of our tableau are closed (whence, $\phi$ is unsatisfiable). Finally, recall that the depth of the model is bounded from above by $\lmc(\phi)$; moreover, each state~$w$ contains only $\Bmc^\Imbb_w$ and $\Smsf_w$, whose sizes are polynomial w.r.t.\ $\lmc(\phi)$. Thus, we need polynomial space to execute the procedure. The result follows.
\end{proof}
\section{Proofs of Section~\ref{sec:polynomial}}
Before proceeding to the proofs, let us introduce the following notational convention.
\begin{convention}\label{conv:RboxRlozenge}
Given a~\emph{sequence} of box-like modalities, $\Rmsf_\blacksquare$ is the \emph{composition} of relations that correspond to modalities in that sequence. Similarly, if $\blacklozenge$~is a~sequence of diamond-like modalities, $\Rmsf_\blacklozenge$ is the composition of relations corresponding to modalities in~$\blacklozenge$.  E.g., if $\blacksquare=\ActBox{a}\knows_B\ActBox{a}$, then $\Rmsf_\blacksquare=\Rmsf_a\circ\Emsf_B\circ\Rmsf_a$. Note that $\Rmsf_\blacksquare=\Rmsf_\blacklozenge$ when modalities in~$\blacksquare$ and~$\blacklozenge$ are over the same sequence of relations. Continuing our example, setting $\blacklozenge=\ActLozenge{a}\knowsLozenge{B}\ActLozenge{b}$ implies $\Rmsf_\blacklozenge=\Rmsf_a\circ\Emsf_B\circ\Rmsf_a=\Rmsf_\blacksquare$.
\end{convention}
\subsection{Proofs of Section~\ref{ssec:UPR}}
In this section, we will prove Theorem~\ref{theorem:UPRpolynomial}. Our plan is as follows. First, we will prove the result for the \emph{outer-layer counterparts} of~$\uprobrule$'s. Then, we will lift it to the probabilistic case. Let us now introduce \emph{universal modal rules}.
\begin{definition}[Universal modal rules]\label{def:UMR}
A \emph{universal modal rule} ($\umodrule$) is
\begin{itemize}[noitemsep,topsep=0pt]
\item either a~formula of the form $\bigodot^n_{i=1}\nu_i\rightarrow\nu'$ s.t.\ $\nu$'s are \emph{proper} $\blacksquare$-literals;
\item or a~formula of the form $\bigodot^n_{i=1}\nu_i\rightarrow\zero$ s.t.\ $\nu$'s are $\blacksquare$-literals.
\end{itemize}
A~\emph{$\umodrule$ theory} is a~set of $\Lmod$-formulas containing $\umodrule$'s. For simplicity, we will write rules of the form $\nu\rightarrow\zero$ as $\blacklozenge$-literals.
\end{definition}
We abuse the notation slightly and assume that $\umodrule$'s are formulated in~$\Lmod$ and interpreted over~$\EALukProd$-models. Furthermore, we observe from Remark~\ref{rem:fewermodels} that $\umodrule$'s have the witnessed model property. Notice also that $\umodrule$'s can be equivalently represented as follows:
\begin{align*}
\bigodot^n_{i=1}\nu_i\rightarrow\nu'&\equiv\bigoplus^n_{i=1}\pi_i\oplus\nu'&\bigodot^n_{i=1}\nu_i\rightarrow\zero&\equiv\bigoplus^n_{i=1}\pi_i\tag{$\neg\nu_i\equiv\pi_i$}
\end{align*}
$\umodrule$'s represented with $\oplus$ are called \emph{clauses}.

Let us now show that the satisfiability of sets of $\uml$'s and $\upl$'s is polynomially decidable. In addition, we establish that the entailment of a~$\uml$ from a~set of~$\uml$'s (and likewise, the entailment of a~$\upl$ from a~set of~$\upl$'s) can also be verified in polynomial time.
\begin{lemma}\label{lemma:UMtermpolynomial}
Given a~set of $\uml$'s $\Lambda$ and a~$\uml$~$\lambda'$, it takes polynomial time to verify whether
\begin{enumerate}[noitemsep,topsep=0pt]
\item $\Lambda$ is $\EALukProd$-satisfiable;
\item $\Lambda\models_{\EALukProd}\lambda'$.
\end{enumerate}
\end{lemma}
\begin{proof}
For the first item, we construct a~tableau for $\{w{:}\lambda{\geq}y,y{=}1\mid\lambda\in\Lambda\}$ (we rewrite $\blacklozenge l$ as $\neg\blacksquare\neg l$). We apply the rules in a~standard manner: first propositional, then rules that generate new states, and then rules that use the states on the branch. As rules for modalities and~$\neg$ do not produce branching, the tableau will contain only one branch~$\Bmc$. As one $\uml$ contains either $\blacksquare$'s or $\blacklozenge$'s but not both, $\Bmc$ will contain at most as many states as there were $\blacklozenge$'s in~$\Lambda$, and each state will contain only subformulas of the formulas in~$\Lambda$. Thus, $\Bmc^\Imbb$ will be of polynomial size w.r.t.\ the size of~$\Lambda$ and will only contain linear inequalities. Thus, it can be solved in polynomial time. For the second item, the reasoning is similar, but we construct a~tableau for $\{w{:}\lambda{\geq}y,y{=}1\mid\lambda\in\Lambda\}\cup\{w{:}\lambda'{\leq}z,z{<}1\}$.
\end{proof}

\begin{lemma}\label{lemma:UPLentailmentpolynomial}
Given a~set of $\upl$'s $\Theta$ and a~$\upl$~$\theta'$, it takes polynomial time to verify whether
\begin{enumerate}[noitemsep,topsep=0pt]
\item $\Theta$ is $\prea$-satisfiable;
\item $\Theta\models_{\prea}\theta'$.
\end{enumerate}
\end{lemma}
\begin{proof}
We begin with Item~1. Let $\Theta=\{\theta_1,\ldots,\theta_n\}$. We construct a~tableau for $\Theta^\Tmc=\{w{:}\theta{\geq}y,y{=}1\mid\theta\in\Theta\}$. From Lemma~\ref{lemma:UMtermpolynomial}, we know that it will take polynomial time to build it and check whether it is closed. Furthermore, one can see that the tableau will have only one branch~$\Bmc$ and (if open) will induce a~model that evaluates probabilistic atoms over~$\{0,1\}$. If the tableau is open, we proceed as follows: for every state $w'$ occurring on~$\Bmc$ and for every agent~$A$ s.t.\ $w'{:}\Prob_A(\lambda)\triangledown P$ occurs on~$\Bmc$, construct a~pointed event model $\langle\Mfrak_{w',A},x_{w',A}\rangle$ for the following set of formulas:
\begin{align}\label{equ:zeroonemodel}
\Lambda_{w',A}&=\{\lambda\mid w'{:}\Prob_A(\lambda)\geq1\text{ occurs on~}\Bmc\}\cup\{\neg\lambda'\mid w'{:}\Prob_A(\lambda')\leq0\text{ occurs on~}\Bmc\}
\end{align}
Note from Fig.~\ref{fig:tableauxrules} that we can reduce polynomials occurring on~$\Bmc$ to either~$1$ or~$0$ in polynomial time because $c=1$ (by the construction of the tableau) is the only constant that can occur on~$\Bmc$, and we only use the $1-P$ operation. Note furthermore, that if $\Lambda_{w',A}$ is classically unsatisfiable, then $\Theta$ is $\prea$-unsatisfiable.

Since $\Lambda_{w',A}$'s contain only \emph{uniform} literals, it is clear that it takes polynomial time to construct a~classical event model for them (or verify that no model exists). It now remains to define a~$\prea$-model for $\Theta$. We set $\Mfrak=\langle W,\Emsf,\Rmsf,\Xfrak,\mu,V\rangle$ as follows: $W=\{w'\mid w'\text{ occurs on~}\Bmc\}$; $w'\Emsf_Aw''$ iff $w'{\in}\cluster^A$ and $w''{\in}\cluster^A$ occur on~$\Bmc$ for some epistemic relational label $\cluster^A$ (recall Definition~\ref{def:relationalterms}); $w'\Rmsf_aw''$ iff $w'\Rmsf_aw''$ occurs on~$\Bmc$; $\langle\Xfrak,V\rangle=\biguplus_{\overset{w'\in W}{A\in\Ag}}\Mfrak_{w',A}$; $\mu_{w',A}(\{x\})=1$ if $x=x_{w',A}$ and $\mu(\{x\})=0$, otherwise. It is clear that $\Imc_\Mfrak(\theta,w)=1$ (with $w$ being the state label occurring in~$\Theta^\Tmc$) for all $\theta\in\Theta$.

For Item~2, we construct a~tableau for $\{w{:}\theta{\geq}y,y{=}1\mid\theta\in\Theta\}\cup\{w{:}\theta'{\leq}z,z<1\}$. Again, from Lemma~\ref{lemma:UMtermpolynomial}, it follows that the tableau can be constructed in polynomial time. Observe also that the application of tableaux rules to $w{:}\theta'{\leq}d$ will produce either $s{:}\Prob_{A'}(\lambda'){<}1$ (if $\theta'$ \emph{does not} contain~$\neg$) or $s{:}\Prob_{A'}(\lambda'){>}0$ (if it \emph{does}). Hence, we need to construct pointed event models $\langle\Mfrak'_{w',A},x'_{w',A}\rangle$ for the following sets of formulas for each~$w'$ and~$A$ occurring on~$\Bmc$:
\begin{align*}
\Lambda'_{w',A}&=\Lambda_{w',A}\cup\{\lambda'\mid w'{:}\Prob_A(\lambda'){>}0\text{ occurs on~}\Bmc\}\cup\{\neg\lambda'\mid w'{:}\Prob_A(\lambda'){<}1\text{ occurs on~}\Bmc\}
\end{align*}
Just as in the previous case, we can build models for $\Lambda'_{w',A}$ (or establish that they do not exist) in polynomial time. We note briefly that if there is \emph{no} model for $\Lambda'_{w',A}$, this means that $\Lambda_{w',A}\models_{\sfive+\Kmbf}\lambda'$ ($\Lambda_{w',A}\models_{\sfive+\Kmbf}\neg\lambda'$, respectively). In this case, there cannot be an event model $\Xfrak$ and a~measure $\mu$ on it s.t.\ $\mu(\|\lambda\|_\Xfrak)=1$ for every $\lambda\in\Lambda_{w',A}$ but $\mu(\|\lambda'\|_\Xfrak)<1$ ($\mu(\|\lambda'\|_\Xfrak)>0$, respectively). Finally, the model witnessing $\Theta\not\models_{\prea}\theta'$ can be built in polynomial time.
\end{proof}

Next, we show that entailment between modal literals can be decided on \emph{classical event models}, i.e., event models s.t.\ $V(p,w)\in\{0,1\}$ for every state and variable. In particular, it will follow that $\Lambda\models_{\EALukProd}\neg\pi$ is equivalent to $\Lambda,\lambda\models_{\EALukProd}\zero$. Note that this is not the case for arbitrary (even propositional) formulas:
\begin{align*}
q,\neg\neg\neg((p\rightarrow\neg p)\rightarrow\neg p)\models_{\EALukProd}\zero&&\text{but}&&q\not\models_{\EALukProd}\neg\neg((p\rightarrow\neg p)\rightarrow\neg p)
\end{align*}
\begin{lemma}\label{lemma:classicalsatisfiabilityreductionlozenge}
Let $\Lambda$ be a~set of $\uml$'s and $\pi$ a~$\blacklozenge$-literal s.t.\ $\Lambda\not\models_{\EALukProd}\neg\pi$. Then $\Lambda,\pi\not\models_{\EALukProd}\zero$ and $\Lambda\cup\{\pi\}$ is satisfiable in a~(classical) event model.
\end{lemma}
\begin{proof}
Let $\Mfrak=\langle W,\Emsf,\Rmsf,V\rangle$ be an $\EALukProd$-model witnessing $\Lambda\not\models_{\EALukProd}\neg\pi$. This means that there is some $w\in W$ s.t.\ $\Imc_\Mfrak(\lambda,w)=1$ for all $\lambda\in\Lambda$ and $\Imc_\Mfrak(\pi,w)>0$. We will define a~(classical) model $\Mfrak'=\langle W,\Emsf,\Rmsf,V'\rangle$ on the same frame and show that $\Imc_{\Mfrak'}(\pi,w)=1$ and $\Imc_{\Mfrak'}(\lambda,w)=1$ for all $\lambda\in\Lambda$. We assume w.l.o.g.\ that $\Mfrak$ is finitely branching and consider two cases: (1)~$\pi$ does not contain~$\neg$ and (2)~$\pi$ contains~$\neg$.

We begin with the first case and proceed by induction on the number of modalities in~$\pi$. Define $V'$ as follows: for all $w{\in}W$, $V'(p,w){=}1$ iff $V(p,w){>}0$. We show that (a) $\Imc_{\Mfrak'}(\lambda,w){=}1$ for every $\neg$-free $\blacklozenge$-literal~$\lambda$ s.t.\ $\Imc_\Mfrak(\lambda,w)>0$ and (b) $\Imc_{\Mfrak'}(\lambda,w)=1$ for every $\uml$~$\lambda$ s.t.\ $\Imc_\Mfrak(\lambda,w)=1$. In the basis case for (a) and~(b), $\lambda=p$. The statements follow by the construction of~$V'$. Now let $\lambda=\neg p$. In this case, (a) holds vacuously, and we consider~(b). If $\Imc_\Mfrak(\neg p,w)=1$, then $V(p,w)=0$, whence $V'(p,w)=0$ and thus, $\Imc_{\Mfrak'}(\neg p,w)=1$, as required. Assume now that $\lambda=\ActLozenge{a}\pi'$ for some \emph{$\neg$-free} $\pi'$. For~(a), let $\Imc_\Mfrak(\ActLozenge{a}\pi',w)>0$. Then $\Imc_\Mfrak(\pi',w')>0$ for some $w'\in \Rmsf_a(w)$, whence by the induction hypothesis, we have that $\Imc_{\Mfrak'}(\pi',w')=1$, and thus, $\Imc_{\Mfrak'}(\ActLozenge{a}\pi',w)=1$. The reasoning for~(b) is similar.

Now let $\lambda=\ActLozenge{a}\pi'$ with $\pi'$ \emph{containing} negation. Again, (a) holds vacuously, so we consider~(b). If $\Imc_\Mfrak(\ActLozenge{a}\pi',w)=1$, then $\Imc_\Mfrak(\pi',w')=1$ for some $w'\in \Rmsf_a(w)$. By the induction hypothesis, we have $\Imc_{\Mfrak'}(\pi',w')=1$, whence $\Imc_{\Mfrak'}(\ActLozenge{a}\pi',w)=1$. The case $\lambda=\knowsLozenge{A}\pi$ can be considered in the same manner.

Finally, let $\lambda=\ActBox{a}\nu'$. We consider~(b). If $\Imc_\Mfrak(\ActBox{a}\nu',w)=1$, then $\Imc_\Mfrak(\nu',w')=1$ for all $w'\in \Rmsf_a(w)$. Hence, by the induction hypothesis, we have $\Imc_{\Mfrak'}(\nu',w')=1$ for all $w'\in \Rmsf_a(w)$. It follows that $\Imc_{\Mfrak'}(\ActBox{a}\nu',w)=1$. The case of $\lambda=\knows_A\nu'$ is similar.

Now consider the second case. Again, we prove the statement by induction on the number of modalities in~$\pi$. We define $V'$ as follows: $V'(p,w)=1$ iff $V(p,w)=1$. Just as in the previous case, we show that $\Imc_{\Mfrak'}(\lambda,w)=1$ for every $\uml$~$\lambda$ s.t.\ $\Imc_\Mfrak(\lambda,w)>0$. The basis case of $\lambda=p$ follows by construction. If $\lambda=\neg p$, then $\Imc_\Mfrak(\neg p,w)>0$ entails that $V(p,w)<1$, whence $V'(p,w)=0$ and $\Imc_{\Mfrak'}(\neg p,w)=1$, as required. If $\lambda=\ActLozenge{a}\pi'$, then $\Imc_\Mfrak(\ActLozenge{a}\pi',w)>0$ implies $\Imc_\Mfrak(\pi',w')>0$ for some $w'\in \Rmsf_a(w)$. By the induction hypothesis, we obtain $\Imc_{\Mfrak'}(\pi',w'){=}1$ for some $w'{\in}\Rmsf_a(w)$, whence \mbox{$\Imc_{\Mfrak'}(\ActLozenge{a}\pi',w)=1$.} The cases of other modalities can be handled in the same way.

It now follows that if $\Mfrak$ witnesses $\Lambda\not\models_{\EALukProd}\neg\pi$, then $\Mfrak'$ witnesses $\Lambda,\pi\not\models_{\EALukProd}\zero$.
\end{proof}

Now, we can establish our first key result.
\begin{definition}\label{def:KLukirreducibleset}
A~$\umodrule$ theory $\Gamma=\Xi\cup\Lambda$ with $\Xi=\{\kappa_1,\ldots,\kappa_m\}$ and $\Lambda=\{\lambda_1,\ldots,\lambda_n\}$ is called \emph{$\EALukProd$-irreducible} if there is no $\lambda$ occurring in~$\Xi$ s.t.\ $\Lambda\models_{\EALukProd}\neg\lambda$.
\end{definition}

Observe from Lemmas~\ref{lemma:UPLentailmentpolynomial} and~\ref{lemma:classicalsatisfiabilityreductionlozenge} that it takes polynomial time to check whether a~given $\umodrule$ theory is $\EALukProd$-irreducible.
\begin{lemma}\label{lemma:irreduciblesatisfiabilityuniversal}
Let $\Gamma=\Xi\cup\Lambda$ with $\Xi=\{\kappa_1,\ldots,\kappa_m\}$ and $\Lambda=\{\lambda_1,\ldots,\lambda_n\}$ be an $\EALukProd$-irreducible $\umodrule$ theory. Then it is $\EALukProd$-satisfiable.
\end{lemma}
\begin{proof}
Assume w.l.o.g.\ that $\kappa$'s contain at least two literals each and $\kappa$'s are written using $\oplus$'s. First, observe that $\Lambda$ is $\EALukProd$-satisfiable: otherwise, $\Lambda\models_{\EALukProd}\neg\lambda$ for every literal~$\lambda$. Now let $\langle\Mfrak,w\rangle$ be a~pointed model s.t.\ $\Imc_\Mfrak(\lambda,w)=1$ for all $\lambda\in\Lambda$. Observe from the proof of Lemma~\ref{lemma:UMtermpolynomial} that since we use tableaux to build the model of~$\Lambda$, it is \emph{classical}, that is, all variables have values in~$\{0,1\}$. Now, in every clause $\kappa\in\Xi$, pick one \emph{proper} $\blacklozenge$-literal $\pi_\kappa$.

Note that there is no $\blacklozenge$-literal $\pi_\kappa$ s.t.\ $\Lambda\models_{\EALukProd}\neg\pi_\kappa$. Thus, by Lemma~\ref{lemma:classicalsatisfiabilityreductionlozenge}, there is a~classical model satisfying $\Lambda\cup\{\pi_\kappa\}$ for each~$\pi_\kappa$. Moreover, it is clear that $\{\pi_\kappa\mid\kappa\in\Xi\}$ is also $\EA$-satisfiable. Let us now show that $\Lambda'=\Lambda\cup\{\pi_\kappa\mid\kappa\in\Xi\}$ is $\EA$-satisfiable.

We proceed as follows. Let $\pi_{\kappa_i}=\blacklozenge_il_i$ for $i\in\{1,\ldots,m\}$ with $l$'s being propositional variables or their negations. Consider a~(classical) pointed event model $\langle\Mfrak,w\rangle$ that $\EA$-satisfies $\Lambda\cup\{\blacklozenge_1l_1\}$. We now need to build an event model that $\EA$-satisfies $\Lambda\cup\{\pi_{\kappa_1},\pi_{\kappa_2}\}$. Since $\Lambda\not\models_{\EALukProd}\neg\blacklozenge_2l_2$, we can w.l.o.g.\ assume that if $\Rmsf_{\blacklozenge_2}(w)\neq\varnothing$, then there is at least one $w'\in\Rmsf_{\blacklozenge_2}$ s.t.\ $V(l_2,w')\neq0$. Since the model is classical, we have $V(l_2,w')=1$ for that state, which gives $\Imc_\Mfrak(\blacklozenge_2l_2,w')=1$, as required. If $\Rmsf_{\blacklozenge_2}(w)=\varnothing$, let $\blacklozenge_2=\blacklozenge_{2,1}\ldots\blacklozenge_{2,k}$ and add states $w_1$, \ldots, $w_k$ to~$\Mfrak$ s.t.\ $w_i\Rmsf_{\blacklozenge_{2,i}}w_{i+1}$ for $i\leq k-1$. For every $w_i\in\{w_1,\ldots,w_k\}$ s.t.\ there is some $\blacksquare'l'\in\Lambda$ with $w_i\in\Rmsf_{\blacksquare'}(w)$, additionally set $V(l',w_i)=1$. Repeating this construction for all $\pi_\kappa$'s we obtain a~classical pointed event model that $\EA$-satisfies $\Lambda'$. Call it  $\langle\Mfrak_\Xi,w\rangle$.

Now, observe that $\langle\Mfrak_\Xi,w\rangle$ satisfies~$\Lambda$ and every clause that contains at least one proper $\blacklozenge$-literal. Thus, the only clauses that are not yet satisfied by~$\langle\Mfrak_\Xi,w\rangle$ are the ones that contain only propositional literals. For those, we set $V(p,w)=\tfrac{1}{2}$ for each variable~$p$ occurring in such clauses. Note that this does not affect the value of formulas in~$\Lambda$: if $\Lambda\models_{\EALukProd}p$, then we do not need to consider propositional clauses containing~$p$, and if $\Lambda\models_{\EALukProd}\neg p$, then no clause can contain~$p$ because $\Gamma$~is irreducible (and dually for~$\neg p$). Similarly, we can w.l.o.g.\ assume that if some proper $\blacklozenge$-literal $\pi$ is true at~$w$, then its value is witnessed by some $w'\neq w$. Finally, we do not consider the values of $\blacksquare$-literals in clauses. It follows that $\Mfrak_\Xi$ (with the updated valuation of variables at~$w$) satisfies $\Gamma$.
\end{proof}

We can now show that $\EALukProd$-satisfiability of $\umodrule$ theories is polynomially decidable.
\begin{lemma}\label{lemma:UMRpolynomial}
Given a~$\umodrule$ theory~$\Gamma$, it takes polynomial time to verify whether it is $\EALukProd$-satisfiable.
\end{lemma}
\begin{proof}
First, observe briefly that if $\Gamma$ contains only clauses with at least two literals each, then it is satisfiable. Indeed, it suffices to consider a~model with one state $w$ s.t.\ all accessibility relations are equal to $\{\langle w,w\rangle\}$ and all variables are evaluated at $\tfrac{1}{2}$. Each literal in every clause will also have value~$\tfrac{1}{2}$. Since there are at least two literals per clause, this satisfies~$\Gamma$.

Now let $\Gamma=\Xi\cup\Lambda$ with $\Xi=\{\kappa_1,\ldots,\kappa_m\}$ and $\Lambda=\{\lambda_1,\ldots,\lambda_n\}$ and $\kappa$'s having at least two literals each. We proceed as follows. First, we check that $\Lambda$ is satisfiable. By Lemma~\ref{lemma:UMtermpolynomial}, this takes us polynomial time. Then for every literal $\lambda$ occurring in~$\Xi$, we check whether (1)~$\Lambda\models_{\EALukProd}\neg\lambda$ or (2)~$\Lambda\models_{\EALukProd}\neg\lambda$. In the first case, we remove from~$\Xi$ all $\kappa$'s that contain~$\lambda$. In the second case, we remove all occurrences of~$\lambda$ from~$\Xi$. Again, this takes us polynomial time to check all literals. If at some stage we have a~clause $\lambda'\oplus\lambda''$ s.t.\ $\Lambda\models_{\EALukProd}\neg\lambda''$, then we remove occurrences of~$\lambda''$ and check that $\Lambda\cup\{\lambda'\}$ is satisfiable. We then use $\Lambda'=\Lambda\cup\{\lambda'\}$ to check the remaining literals. We proceed in this way until either (i) we encounter a~clause $\lambda_1\oplus\lambda_2$ s.t.\ $\Lambda'\models_{\EALukProd}\neg\lambda_1$ and $\Lambda',\lambda_2\models_{\EALukProd}\zero$ or (ii) we produce an $\EALukProd$-irreducible theory. In the first case, $\Gamma$ is unsatisfiable, in the second case, $\Gamma$ is satisfiable.
\end{proof}

It remains to lift our argument to universal \emph{probabilistic} rules. The next statement can be shown similarly to Lemma~\ref{lemma:classicalsatisfiabilityreductionlozenge}.
\begin{lemma}\label{lemma:classicalsatisfiabilityreductionlozengeprob}
Let $\Theta$ be a~set of $\upl$'s and $\zeta$ a~$\blacklozenge^\Prob$-literal s.t.\ $\Theta\not\models_{\prea}\neg\zeta$. Then $\Theta,\zeta\not\models_{\prea}\zero$.
\end{lemma}
\begin{proof}
Let $\zeta=\blacklozenge'\Prob_A(\lambda')$. We construct a~tableau for $\{w{:}\theta\geq y,y=1\mid\theta\in\Theta\}\cup\{w{:}\blacksquare\neg\Prob_A(\lambda')\leq z,z<1\}$. Since $\Theta\not\models_{\prea}\neg\zeta$, this tableau should be open and also have only one branch~$\Bmc$. Now, let $w'$ be the state s.t.\ $w'{:}\Prob_A(\lambda')\geq1-z$ occurs on~$\Bmc$ and that was generated by applying $\Box_\leq$ and $\knows_\leq$ rules to~$\blacksquare\neg\Prob_A(\lambda')$. Observe, furthermore, that apart from $\Prob_A(\lambda')$, all other formulaic constraints labelled with~$w'$ on~$\Bmc$ have the form $w'{:}\phi\geq1$ or $w'{:}\phi\leq0$.

Using the reasoning from Theorem~\ref{theorem:generalcasedecidability}, we can assume w.l.o.g.\ that there is an SI model $\Mfrak=\langle W,\Rmsf,\Emsf,\Xfrak,\mu,V\rangle$ that witnesses $\Theta\not\models_{\prea}\neg\zeta$ s.t.\ $W=\{w\mid w\text{ occurs on }\Bmc\}$ and $\Xfrak=\biguplus_{w\in W}\Xfrak_w$. It follows that $\mu_{w',A}(\|\lambda\|_{\Xfrak_{w'}})=1$ for every $\lambda$ s.t.\ $w'{:}\Prob_A(\lambda)\geq1$ occurs on~$\Bmc$ and $\mu(\|\lambda\|_{\Xfrak_{w'}})=0$ for every $\lambda$ s.t.\ $w'{:}\Prob_A(\lambda)\leq0$ occurs on~$\Bmc$. Additionally, we have that $\mu_{w',A}(\|\lambda'\|_{\Xfrak_{w',A}})=z'>0$. 

Now, we can construct $\Lambda_{w',A}$ as shown in~\eqref{equ:zeroonemodel}. It follows that there is some $x\in\Xfrak_{w'}$ s.t.\ $\Xfrak_{w'},x\vDash\lambda'\wedge\bigwedge_{\lambda''\in\Lambda_{w',A}}\lambda''$. Now, we can replace $\mu_{w',A}$ with $\mu'_{w',A}$ defined as follows: $\mu'_{w',A}(X)=1$ if $x\in X$ and $\mu'_{w',A}(X)=0$, otherwise. Observe that $\mu'_{w',A}(\|\lambda'\|_{\Xfrak_{w'}})=1$. Now let $\Mfrak'$ be $\Mfrak$ where $\mu_{w,A}$ is replaced with $\mu'_{w,A}$. It is clear that $\Imc_{\Mfrak}(\blacklozenge\Prob_A(\lambda'),w)=1$ and $\Imc_{\Mfrak'}(\theta,w)=1$ for all $\theta\in\Theta$. Thus, $\Mfrak'$~witnesses $\Theta,\zeta\not\models_{\prea}\zero$.

Finally, if $\zeta=\blacklozenge'\neg\Prob_A(\lambda')$, the reasoning is similar. The only difference is that we will have $\Xfrak_{w'}$ and $x$ in such a~way that $\Xfrak_{w'},x\vDash\neg\lambda'\wedge\bigwedge_{\lambda''\in\Lambda_{w',A}}\lambda''$. Thus, redefining the measure as above, will give us $\mu'_{w',A}(\|\neg\lambda'\|_{\Xfrak_{w'}})=1$, whence $\mu'_{w',A}(\|\lambda'\|_{\Xfrak_{w'}})=0$, $\Imc_{\Mfrak'}(\neg\Prob_A(\lambda'),w')=1$ and $\Imc_{\Mfrak}(\blacklozenge\neg\Prob_A(\lambda'),w)=1$.
\end{proof}
\begin{definition}\label{def:preairreducibleset}
A~$\uprobrule$ theory $\Gamma=\Theta\cup\Xi$ with $\Theta=\{\theta_1,\ldots,\theta_m\}$ and $\Xi=\{\kappa_1,\ldots,\kappa_n\}$ is called \emph{$\prea$-irreducible} iff there is no $\upl$ $\theta$ occurring in~$\Xi$ s.t.\ $\Theta\models_{\prea}\neg\theta$.
\end{definition}
Just as it was with $\umodrule$ theories, one can see that it takes polynomial time to check that a~given $\uprobrule$ theory is $\prea$-irreducible. Let us now show that irreducible $\uprobrule$ theories are $\prea$-satisfiable.
\begin{lemma}\label{lemma:preairreduciblesatisfiabilityupr}
Let $\Gamma=\Theta\cup\Xi$ with $\Theta=\{\theta_1,\ldots,\theta_m\}$ and $\Xi=\{\kappa_1,\ldots,\kappa_n\}$ be an $\prea$-irreducible $\uprobrule$ theory. Then it is $\prea$-satisfiable.
\end{lemma}
\begin{proof}
We reason in the same manner as for Lemma~\ref{lemma:irreduciblesatisfiabilityuniversal}, but use Lemma~\ref{lemma:classicalsatisfiabilityreductionlozengeprob} instead of Lemma~\ref{lemma:classicalsatisfiabilityreductionlozenge}. For every $\kappa\in\Xi$, we pick one proper $\blacklozenge^\Prob$-literal $\zeta_\kappa$ and consider a~pointed $\prea$-model $\langle\Mfrak,w\rangle$ that satisfies $\Theta'=\Theta\cup\{\zeta_\kappa\mid\kappa\in\Xi\}$. For every clause $\kappa'\in\Xi$ that \emph{does not} contain proper literals, we set its literals at~$\tfrac{1}{2}$. To do that, we consider a~disjoint union of two event models: $\Xfrak^+\uplus\Xfrak^-$ s.t.\ $\Xfrak^+,x^+\vDash\lambda'$ and $\Xfrak^-,x^-\nvDash\lambda'$ for every $\lambda'$ that occurs in \emph{some} $\kappa'$ that does not contain proper literals. Then, for every $A\in\Ag$, we define $\mu_{w,A}(\{x^+\})=\mu_{w,A}(\{x^-\})=\tfrac{1}{2}$. Now, given a~clause $\kappa'$ that does not contain proper literals, we have $\Imc_\Mfrak(\zeta',w)=\tfrac{1}{2}$ for every $\zeta'$ in~$\kappa'$. Thus, $\Mfrak$ satisfies~$\Gamma$.
\end{proof}
Now, using the same reasoning as in Lemma~\ref{lemma:UMRpolynomial} and applying Lemmas~\ref{lemma:UPLentailmentpolynomial} and~\ref{lemma:preairreduciblesatisfiabilityupr}, we can obtain the main result of the section.
\UPRpolynomial*
\subsection{Proofs of Section~\ref{ssec:EPR}}
In this section, we prove Theorem~\ref{theorem:EPRpolynomial}. Our approach is similar to the case of universal probabilistic rules. That is, we first prove our results for the modal Łukasiewicz logic and then expand them to the probabilistic case.
\begin{definition}\label{def:EMR}~
An \emph{existential positive modal rule} ($\emrplus$) is a~formula of one of the following kinds:
\begin{itemize}[noitemsep,topsep=0pt]
\item a~$\uml$;
\item $\bigodot^n_{i=1}\pi_i\rightarrow\zero$ with $\pi$'s being propositional literals or $\neg$-free $\blacklozenge$-literals;
\item $\bigodot^n_{i=1}\pi_i\rightarrow\pi$ with $\pi$'s being \emph{proper} $\neg$-free $\blacklozenge$-literals.
\end{itemize}
\end{definition}

Note that $\emrplus$'s can also be represented as $\oplus$-clauses that contain at most one proper $\blacksquare$-literal. Thus, in a~sense they are dual to $\umodrule$'s (though, they are a~bit more restrictive since they cannot have both~$\neg$ and~$\knows$ in the same theory).

Let us now establish an analogue of Lemma~\ref{lemma:classicalsatisfiabilityreductionlozenge} for $\blacksquare$-literals. The proof is similar.
\begin{lemma}\label{lemma:classicalsatisfiabilityreductionbox}
Let $\Lambda$ be a~set of $\uml$'s and $\nu$ a~$\blacksquare$-literal s.t.\ $\Lambda\not\models_{\EALukProd}\neg\nu$. Then $\Lambda,\nu\not\models_{\EALukProd}\zero$ and $\Lambda\cup\{\nu\}$ is satisfiable in a~(classical) event model.
\end{lemma}
\begin{proof}
Let $\Mfrak=\langle W,\Emsf,\Rmsf,V\rangle$ be an $\EALukProd$-model witnessing $\Lambda\not\models_{\EALukProd}\neg\nu$. This means that there is some $w\in W$ s.t.\ $\Imc_\Mfrak(\lambda,w)=1$ for all $\lambda\in\Lambda$ and $\Imc_\Mfrak(\nu,w)>0$. We will define a~classical model $\Mfrak'=\langle W,\Emsf,\Rmsf,V'\rangle$ on the same frame and show that $\Imc_{\Mfrak'}(\nu,w)=1$ and $\Imc_{\Mfrak'}(\lambda,w)=1$ for all $\lambda\in\Lambda$. We assume w.l.o.g.\ that $\Mfrak$ is finitely branching and consider two cases: (1)~$\nu$ does not contain~$\neg$ and (2)~$\nu$ contains~$\neg$.

In the first case, for all $w\in W$, we define $V'(p,w)=1$ iff $V(p,w)>0$ (and $V'(p,w)=0$, otherwise). In the second case, for all $w\in W$, we set $V'(p,w)=1$ iff $V(p,w)=1$ (and $V'(p,w)=0$, otherwise). We can then show by induction on the number of modalities in~$\nu$ that $\Mfrak'$~satisfies $\Lambda\cup\{\nu\}$. The reasoning is the same as in Lemma~\ref{lemma:classicalsatisfiabilityreductionlozenge}.
\end{proof}
We define $\EALukProd$-irreducible $\emrplus$ theories in the same way as $\umodrule$ theories (recall Definition~\ref{def:KLukirreducibleset}). Let us now show the analogue of Lemma~\ref{lemma:irreduciblesatisfiabilityuniversal}.
\begin{lemma}\label{lemma:irreduciblesatisfiabilityexistential}
Let $\Gamma=\Xi\cup\Lambda$ with $\Xi=\{\kappa_1,\ldots,\kappa_m\}$ and $\Lambda=\{\lambda_1,\ldots,\lambda_n\}$ be an $\EALukProd$-irreducible $\emrplus$ theory. Then it is $\EALukProd$-satisfiable.
\end{lemma}
\begin{proof}
The proof is similar to that of Lemma~\ref{lemma:irreduciblesatisfiabilityuniversal}. The main difference is that we will use Lemma~\ref{lemma:classicalsatisfiabilityreductionbox} and the fact that any set~$\Lambda'$ of proper $\blacksquare$-literals is classically satisfiable if $\nu=\blacksquare\neg p$ for all $\nu\in\Lambda$.

Now, if $\Gamma$~is an $\emrplus$ theory, we use this fact as follows. Since $\Gamma$ is irreducible, there are a~classical model $\Mfrak=\langle W,\Emsf,\Rmsf,V\rangle$ and $w\in W$ s.t.\ $\Mfrak,w\vDash\lambda$ for all $\lambda\in\Lambda$. Now, we can represent rules in~$\Xi$ as $\oplus$-clauses of the following form: $\blacklozenge p\oplus\bigoplus^k_{i=1}\blacksquare\neg p_i$ and $\bigoplus^k_{i=1}\blacksquare\neg p_i$, depending on whether the rules had empty head. From Definition~\ref{def:EMR}, it follows that either (A) a~clause contains a~proper $\blacksquare$-literal or (B) it is propositional. Now for every clause $\kappa$ that satisfies~(A), we choose a~proper $\blacksquare$-literal $\nu_\kappa$.

We show that $\Lambda\cup\{\nu_\kappa\mid\kappa\in\Xi\}$ is $\EA$-satisfiable. For that, construct a~pointed event model $\langle\Mfrak,w\rangle$ that $\EA$-satisfies~$\Lambda$ using tableaux. Observe that $\blacksquare\neg p\in\Theta$ s.t.\ $\Rmsf_\blacksquare(w)\neq\varnothing$, we have $V(p,w')=0$ for all $w'\in\Rmsf_\blacksquare(w)$. Indeed, if there were some state $w''\in\Rmsf_\blacksquare$ s.t.\ $V(p,w)=1$, it would mean that $\Lambda\models_{\EALukProd}\blacklozenge p$, i.e., $\Lambda\models_{\EALukProd}\neg\blacksquare\neg p$. But this would contradict our assumption that $\Gamma$~is irreducible. Hence, $\langle\Mfrak,w\rangle$ satisfies $\Lambda\cup\Theta$. Now, the only clauses in~$\Gamma$ that are not satisfied by~$\langle\Mfrak,w\rangle$ are propositional. We proceed in the same way as in Lemma~\ref{lemma:irreduciblesatisfiabilityuniversal} and set $V(p,w)=\tfrac{1}{2}$ for every variable in every propositional clause that is not satisfied by~$\langle\Mfrak,w\rangle$. It follows that $\Mfrak$ (with an updated valuation) satisfies~$\Gamma$.
\end{proof}

The next statement can now be shown in the same way as Lemma~\ref{lemma:UMRpolynomial}.
\begin{lemma}\label{lemma:EMRpolynomial}
Given an $\emrplus$ theory~$\Gamma$, it takes polynomial time to verify whether it is $\EALukProd$-satisfiable.
\end{lemma}

Now, the following statements can be obtained in the same way as Lemmas~\ref{lemma:classicalsatisfiabilityreductionlozengeprob} and~\ref{lemma:preairreduciblesatisfiabilityupr}.
\begin{lemma}\label{lemma:classicalsatisfiabilityreductionboxprob}
Let $\Theta$ be a~set of $\upl$'s and $\eta$ a~$\blacksquare^\Prob$-literal s.t.\ $\Theta\not\models_{\prea}\neg\zeta$. Then $\Theta,\zeta\not\models_{\prea}\zero$.
\end{lemma}
\begin{lemma}\label{lemma:preairreduciblesatisfiabilityepr}
Let $\Gamma=\Theta\cup\Xi$ with $\Theta=\{\theta_1,\ldots,\theta_m\}$ and $\Xi=\{\kappa_1,\ldots,\kappa_n\}$ be an $\prea$-irreducible $\eprplus$ theory. Then it is $\prea$-satisfiable.
\end{lemma}

Finally, the main result of the section can be obtained from Lemmas~\ref{lemma:UPLentailmentpolynomial}, \ref{lemma:classicalsatisfiabilityreductionlozengeprob}, and~\ref{lemma:preairreduciblesatisfiabilityepr} in the same way as Theorem~\ref{theorem:UPRpolynomial}.
\EPRpolynomial*
\end{document}